\documentclass[10pt,journal]{IEEEtran}
\usepackage{amsmath}
\usepackage{amsthm}
\usepackage{amsfonts}
\usepackage{amssymb}
\usepackage{amscd}
\usepackage{graphicx}
\usepackage{newlfont}
\usepackage{cite}
\usepackage{color}
\usepackage{multirow,tabularx}
\usepackage{ifthen}
\usepackage{enumitem}
\usepackage{times}
\usepackage{stackrel}
\usepackage{float}
\usepackage[all]{xy}
\usepackage{url}
\usepackage{booktabs}
\usepackage{multirow,tabularx}
\usepackage{hyphenat}
\usepackage{flushend}
\usepackage{soul}

\def\BibTeX{{\rm B\kern-.05em{\sc i\kern-.025em b}\kern-.08em
    T\kern-.1667em\lower.7ex\hbox{E}\kern-.125emX}}

\IEEEoverridecommandlockouts

\newcommand{\lb}{\left(}
\newcommand{\rb}{\right)}

\newtheorem{theorem}{{Theorem}}
\newtheorem{lemma}{{Lemma}}

\newtheorem{proposition}[theorem]{Proposition}
\newtheorem{corollary}{{\bf Corollary}}

\newtheorem{defn}{Definition}

\newcommand{\diag}{\mathrm{diag}}

\newcommand{\Real}{\mathbb{R}}
\newcommand{\Complex}{\mathbb{C}}

\newcommand{\ls}{\left[}
\newcommand{\rs}{\right]}
\newcommand{\lc}{\left\{}
\newcommand{\rc}{\right\}}
\newcommand{\lv}{\left|}
\newcommand{\rv}{\right|}
\newcommand{\lvv}{\left| \left|}
\newcommand{\rvv}{\right| \right|}
\newcommand{\lvvv}{\left| \left| \left|}
\newcommand{\rvvv}{\right| \right| \right|}
\newcommand{\veca}{\mathbf{a}}
\newcommand{\vecc}{\mathbf{c}}

\newcommand{\vecp}{\mathbf{p}}

\newcommand{\vecn}{\mathbf{n}}

\newcommand{\vecx}{\mathbf{x}}
\newcommand{\vecz}{\mathbf{z}}
\newcommand{\vecy}{\mathbf{y}}
\newcommand{\vecv}{\mathbf{v}}
\newcommand{\vecu}{\mathbf{u}}
\newcommand{\vecb}{\mathbf{b}}
\newcommand{\vecd}{\mathbf{d}}
\newcommand{\vecw}{\mathbf{w}}
\newcommand{\matA}{\mathbf{A}}
\newcommand{\matB}{\mathbf{B}}
\newcommand{\matI}{\mathbf{I}}
\newcommand{\matC}{\mathbf{C}}
\newcommand{\matD}{\mathbf{D}}

\newcommand{\matH}{\mathbf{H}}
\newcommand{\matJ}{\mathbf{J}}

\newcommand{\matT}{\mathbf{T}}
\newcommand{\matP}{\mathbf{P}}
\newcommand{\matV}{\mathbf{V}}
\newcommand{\matU}{\mathbf{U}}

\newcommand{\matR}{\mathbf{R}}
\newcommand{\matX}{\mathbf{X}}

\newcommand{\matLambda}{\mathbf{\Lambda}}
\newcommand{\matPhi}{\mathbf{\Phi}}

\newcommand{\vgamma}{\boldsymbol{\boldsymbol{\boldsymbol{\gamma}}}}

\newcommand{\vecyall}{\vecy_{1}, \vecy_{2}, \dots, \vecy_{L}}
\newcommand{\vecxall}{\vecx_{1}, \vecx_{2}, \dots, \vecx_{L}}
\newcommand{\EE}{\mathbb{E}}
\newcommand{\PP}{\mathbb{P}}
\newcommand{\setS}{\mathcal{S}}

\newcommand{\setU}{\mathcal{U}}

\newcommand{\AhB}{\mathbf{A} \circ \mathbf{B}}
\newcommand{\AkB}{\mathbf{A} \otimes \mathbf{B}}
\newcommand{\AkrB}{\mathbf{A} \odot \mathbf{B}}

\allowdisplaybreaks

\begin{document}

\title{On the Restricted Isometry of the Columnwise Khatri-Rao Product}
\author{\authorblockN{Saurabh Khanna, \emph{Student Member, IEEE} and Chandra R. Murthy, \emph{Senior Member, IEEE}}\\
\authorblockA{\begin{tabular}{cc}
Dept. of ECE, Indian Institute of Science,
Bangalore, India \\
\{saurabh,~cmurthy\}@\textcolor{black}{iisc.ac.in} \\
\end{tabular}}
\vspace{-5mm}
}


\maketitle

\begin{abstract}
The columnwise Khatri-Rao product of two matrices is an important matrix type, 
reprising \textcolor{black}{its role} as a structured sensing matrix in many fundamental linear inverse problems.
Robust signal recovery in such inverse problems  
is often contingent on proving the restricted isometry property (RIP) of a \textcolor{black}{certain system} matrix 
expressible as a Khatri-Rao product of two matrices. 
In this work, we \textcolor{black}{analyze} the RIP of a generic columnwise 
Khatri-Rao product matrix by deriving two upper bounds for 
its $k^{\text{th}}$ order Restricted Isometry Constant ($k$-RIC) 
for different values of $k$.
The first RIC bound is computed in terms of the individual RICs of the input matrices participating in the Khatri-Rao product. 
The second RIC bound is probabilistic, and is specified in terms of the input matrix dimensions.
We show that the Khatri-Rao product of a pair of $m \times n$ sized random matrices 
comprising independent and identically 
distributed subgaussian entries satisfies $k$-RIP with arbitrarily high probability, 
provided $m$ exceeds $O(k \log n)$. 
Our RIC bounds confirm that the Khatri-Rao product exhibits stronger restricted isometry compared to its constituent matrices 
for the same RIP order. The proposed RIC bounds are potentially useful in the sample complexity analysis of several sparse recovery problems. 

\end{abstract}
\begin{keywords}
Khatri-Rao product, Kronecker product, compressive sensing, 
Restricted isometry property, covariance matrix estimation, 
\textcolor{black}{multiple measurement vectors, PARAFAC, CANDECOMP, direction of arrival estimation.}
\end{keywords}

\section{Introduction}
The Khatri-Rao product, denoted by the symbol $\odot$, is a columnwise Kronecker product, which was originally introduced by 
Khatri and Rao in \cite{KhatriRao68}. For any two matrices $\matA = \ls \veca_{1}, \veca_{2} \dots, \veca_{p} \rs$ and 
$\matB = \ls \vecb_{1}, \vecb_{2} \dots, \vecb_{p} \rs$ of sizes 
$m \times p$ and $n \times p$, respectively, the columnwise Khatri-Rao product $\matA \odot \matB$ is a matrix of dimension $mn \times p$ defined as
\begin{equation} \label{defn_colwise_krprod}
\AkrB = \ls \veca_{1} \otimes \vecb_{1}  \;\; 
\veca_{2} \otimes \vecb_{2} \;\;\;\;
\dots \;\;\;\;
\veca_{p} \otimes \vecb_{p}
\rs,
\end{equation}
where $\veca \otimes \vecb$ denotes the Kronecker product \cite{Bernstein09MatrixMath} between vectors $\veca$ and $\vecb$. 
That is, each column of $\AkrB$ is the Kronecker product between the respective columns of the two input matrices $\matA$ and $\matB$.  
In this article, we shall refer to the columnwise Khatri-Rao product as simply the Khatri-Rao product 
or the KR product. 
Since the Kronecker product $\AkB$ comprises all pairwise Kronecker product combinations of the columns 
of the input matrices, it can be shown that $\AkrB = (\AkB) \matJ$, where $\matJ$ is a $p^{2} \times p$ selection matrix 
with columns as a subset of the standard basis in $\Real^{p^2}$ \cite{LiuTrenkler}. 

Khatri-Rao product matrices are encountered in several linear inverse problems of fundamental importance. 
Recent examples include compressive sensing \cite{Duarte12KroneckerCS,  PPal15SrcLocalization}, covariance matrix estimation \cite{Romero16CovSense, Dasarathy15CovMatEst}, 
direction of arrival estimation \cite{Ma10DoaEstViaKR} and tensor decomposition \cite{Sidiropoulos12CSTensors}. 
In these examples, the KR product $\matA \odot \matB$, for certain $m \times n$ sized system matrices $\matA$ and $\matB$, plays the role of the sensing matrix used to generate linear measurements $\vecy$ of an unknown signal vector $\vecx$ according to 
\begin{equation} \label{linear_measurement_model}
\vecy = \lb \AkrB \rb \vecx + \vecw, 
\end{equation}
where $\vecw$ represents the additive measurement noise. 
It is now well established in the sparse signal recovery 
literature~\cite{Donoho06StableRecovery, FoucartRauhut13CSBook, CandesTaoRomberg06StableRecovery} 
that, if the signal of interest, $\vecx$, is a $k$-sparse\footnote{A vector is said to be $k$-sparse 
if at most $k$ of its entries are nonzero.} vector in $\Real^{n}$, 
it can be stably recovered from its noisy underdetermined linear observations $\vecy \in \Real^{m^{2}} (m^2 < n)$ in a computationally efficient manner 
provided that the sensing matrix (here, $\AkrB$) satisfies the restricted isometry property defined next. 

A matrix $\matPhi \in \Real^{m \times n}$ is said to satisfy the \textit{Restricted Isometry Property} (RIP) \cite{CandesTao05} of order $k$, 
if there exists a constant $\delta_{k}(\matPhi) \in \lb 0, 1 \rb$, 
such that for all $k$-sparse vectors $\vecz \in \Real^{n}$, 
\begin{equation} \label{defn_rip}
(1 - \delta_{k}(\matPhi))|| \vecz ||_{2}^{2}
\le ||\matPhi \vecz ||^{2}_{2} \le 
(1 + \delta_{k}(\matPhi))|| \vecz ||_{2}^{2}.
\end{equation}
The smallest constant $\delta_{k}(\matPhi)$ for which \eqref{defn_rip} holds for all $k$-sparse $\vecz$ is called the 
$k^{\text{th}}$ order restricted isometry constant or the $k$-RIC of $\matPhi$. 
Matrices with small $k$-RICs are good encoders for storing/sketching high dimensional vectors with $k$ or fewer nonzero entries 
\cite{Baraniuk08ASimpleProof}. 
For example, $\delta_{k}(\AkrB) < 0.307$ is a sufficient condition 
for a unique $k$-sparse solution to \eqref{linear_measurement_model} in the noiseless case, and its perfect recovery via the $\ell_{1}$ minimization technique \cite{Cai10RICbound}. 
As pointed out earlier, in many structured signal recovery problems, the main sensing matrix 
can be expressed as a columnwise Khatri-Rao product between two matrices. 
Thus, from a practitioner's viewpoint, it is pertinent to study the 
restricted isometry property of a columnwise Khatri-Rao product matrix, 
which is the focus of this work.


\subsection{Applications involving Khatri-Rao matrices}
We briefly discuss some examples where it is required to 
show the restricted isometry property of a KR product matrix.

\subsubsection{\textbf{\textcolor{black}{Support recovery of joint sparse vectors from underdetermined linear measurements}}} 

Suppose $\vecxall$ are unknown joint sparse signals in $\Real^{n}$ with a common 
$k$-sized support denoted by an index set $\setS$. A canonical problem in multi-sensor 
signal processing is concerned with the recovery of the common support $\setS$ of the unknown 
signals from their noisy underdetermined linear measurements $\vecyall \in \Real^{m}$ 
generated according to 
\begin{equation} \label{mmv_model}
\vecy_{j} = \matA \vecx_{j} + \vecw_{j}, \;\;\;\; 1 \le j \le L,
\end{equation}
where $\matA \in \Real^{m \times n} (m < n)$ is a known measurement matrix, 
and $\vecw_{j} \in \Real^{n}$ models the noise in the measurements.
This problem arises in many practical applications such as 
MIMO channel estimation, cooperative wideband spectrum sensing in cognitive radio networks, 
target localization, and direction of arrival estimation. 
In \cite{PPal15CovMMV}, the support set $\setS$ is recovered as the 
support of $\hat{\vgamma}$, the solution to the Co-LASSO problem: 
\begin{equation} \label{colasso_problem}
\text{Co-LASSO: } \min_{\vgamma \succeq 0} \lvv \text{vec} (\hat{\matR}_{\vecy\vecy}) - (\matA \odot \matA) \vgamma \rvv_{2}^{2} + \lambda \lvv \vgamma \rvv_{1}, 
\end{equation}
where $\hat{\matR}_{\vecy\vecy} \triangleq \frac{1}{L} \sum_{j=1}^{L}\vecy_{j}\vecy_{j}^{T}$. 
From compressive sensing theory \cite{FoucartRauhut13CSBook}, 
the RIP of $\matA \odot \matA$ (also called the self Khatri-Rao product of $\matA$) 
determines the stability of the sparse solution in the Co-LASSO problem. 

In M-SBL \cite{WipfRao07MSBL}, a different support recovery algorithm, $\matA \odot \matA$ satisfying 
$2k$-RIP can guarantee exact recovery of $\setS$ 
from multiple measurements~\cite{KhannaCRM17MSBLSuffCond}.

\subsubsection{\textbf{Sparse sampling of stationary graph signals}}
Let $\vecx = \ls \vecx_{1}, \vecx_{2}, \ldots, \vecx_{n}\rs^{T} \in \Complex^{n}$ be a stochastic, 
zero mean and second-order stationary graph signal defined on $n$ vertices of a graph $\mathcal{G}$. 
This implies that the graph signal $\vecx$ can be modeled as $\vecx = \matH \vecn$, where $\matH$ is 
any valid graph filter \cite{Chepuri17GraphSamp}, and $n \sim \mathcal{N}(0, \matI_{n})$. The covariance matrix $\matR_{\vecx\vecx} = \mathbb{E}[\vecx \vecx^{H}]$ can then be expressed as 
\begin{eqnarray}
\matR_{\vecx\vecx} &=& \matH \mathbb{E}[\vecn \vecn^{H}] \matH^{H} 
= \matH \matH^{H}
\nonumber \\
&=& \matU \diag{(\vecp)} \matU^{H},
\end{eqnarray}  
where the nonnegative vector $\vecp$ refers to the graph power spectral density of the stationary graph signal $\vecx$ and the columns of $\matU$ serve as Fourier like orthonormal basis for the graph signal.

As motivated in \cite{Chepuri17GraphSamp}, in many applicaions, we are interested in reconstructing the sparse graph power spectral density $\vecp$ by observing a small subset of the graph vertices. Let $\vecyall$ denote the $L$ independent obervations of the subsampled graph signal $\vecx$, i.e.,
\begin{equation}
\vecy_{j} = \matPhi \vecx_{j},  \;\;\;\; 1 \le j \le L,
\end{equation} 
where $\matPhi \in \lc 0,1 \rc^{m \times n}$ is referred to as a binary subsampling matrix with $m (\ll n)$ rows, each row containing exactly one nonzero unity element. To recover the graph power spectral density $\vecp$ from the subsampled observations, we note that 
\begin{eqnarray}
\frac{1}{L}\sum_{j = 1}^{L} \vecy_{j}\vecy_{j}^{H} \!\!\!&\approx& \!\!\! \matPhi \matR_{\vecx \vecx}\matPhi^{T} 
= \matPhi \matU \diag{(\vecp)} \matU^{H} \matPhi^{T} 
\nonumber \\
\text{or, } \text{vec}\lb \frac{1}{L}\sum_{j = 1}^{L} \vecy_{j}\vecy_{j}^{H} \rb 
\!\!\!&\approx & \!\!\! (\matPhi \matU^{*} \odot \matPhi \matU) \vecp,
\label{eqn_stationary_graph_signal_example}
\end{eqnarray} 
the superscript $*$ denoting conjugation without transpose.
Here again, the Khatri-Rao product $\matPhi \matU^{*} \odot \matPhi \matU$ satisfying RIP of order~$2r$ 
guarantees a unique $r$-sparse solution for $\vecp$ in~\eqref{eqn_stationary_graph_signal_example}.

\ifdefined \SKIPTEMP
\subsubsection{\textbf{\textcolor{black}{Vandermonde decomposition of Toeplitz matrices}}} 
According to a classical result by Carath\'{e}odory and Fej\'{e}r \cite{caratheodory1911uber}, 
any $n \times n$ positive semidefinite Toeplitz matrix $\matT$ of rank $r < n$ admits 
the following decomposition 
\begin{equation} \label{Toeplitz_decomposition}
\matT = \matA \matP \matA^{H},
\end{equation}
where $\matP$ is an $n \times n$ positive semidefinite diagonal matrix with an 
$r$-sparse diagonal, and $\matA$ is an $n \times n$ Vandermonde 
matrix with uniformly sampled complex sinusoids of different frequencies as its columns. 
This Toeplitz decomposition underpins subspace based 
spectrum estimation methods such as MUltiple SIgnal Classification (MUSIC) and 
EStimation of Parameters by Rotationally Invariant Techniques (ESPRIT) 
\cite{stoica2005spectral, YangStoica16Toeplitz}.
By replacing $\matT$ with a data covariance matrix, the $r$-sparse support of 
$\diag{(\matP)}$ in \eqref{Toeplitz_decomposition} 
corresponds to the $r$-dimensional signal subspace of the data. 
Estimation of $\vecp \triangleq \diag{(\matP)}$ is 
tantamount to finding an $r$-sparse solution to the vectorized form 
of \eqref{Toeplitz_decomposition}, i.e.,
$\text{vec}(\matT) = (\matA^{*} \odot \matA) \vecp$, where superscript $*$ denotes 
conjugation without transpose.
Here again, the recovery of a unique $r$-sparse solution for $\vecp$ 
can be guaranteed if $\matA^{*} \odot \matA$ satisfies the RIP of order~$2r$.
\fi

\subsubsection{\textbf{\textcolor{black}{PARAFAC model for low-rank three-way arrays}}}
\textcolor{black}{
Consider an $I \times J \times K$ tensor $\underline{\matX}$ of rank $r$. 
We can express $\underline{\matX}$ as the sum of $r$ rank-one three 
way arrays as $\underline{\matX} = \sum_{i=1}^{r} \veca_{i} \circ \vecb_{i} \circ \vecc_{i}$, 
where $\veca_{i}, \vecb_{i}, \vecc_{i}$ are  loading vectors of dimension $I, J, K$, respectively, 
and $\circ$ denotes the vector outer product. The tensor $\underline{\matX}$ itself can be 
arranged into a matrix as 
$\matX = \ls \text{vec}(\matX_{1}), \text{vec}(\matX_{2}), \ldots, \text{vec}(\matX_{K})\rs$.  
In the parallel factor analysis (PARAFAC) model \cite{Sidiropoulous12PARAFAC}, 
the matrix $\matX$ can be approximated as 
\begin{equation} \label{PARAFAC}
\matX \approx (\matA \odot \matB) \matC^{T},
\end{equation}
where $\matA, \matB$ and $\matC$ are the loading matrices with columns as the loading vectors 
$\veca_{i}$, $\vecb_{i}$ and $\vecc_{i}$, respectively. In many problems such as 
direction of arrival estimation using a 2D-antenna array, the loading matrix $\matC$ 
turns out to be row-sparse matrix \cite{Sidiropolos09MIMORadar}. In such cases, 
the uniqueness of the PARAFAC model shown in \eqref{PARAFAC} depends on the 
restricted isometry property of the Khatri-Rao product $\matA \odot \matB$. 
}

Finding the exact $k{\text{th}}$ order RIC of a given matrix $\matX$ entails computation of extreme singular values of all possible $k$-column submatrices of $\matX$, which is an NP hard task \cite{Pfetsch14RIPCalcNPHard}. 
In this work, we follow an alternative approach to analyzing the RIP of a KR product matrix. We seek to derive tight upper bounds for its restricted isometry constants.

\subsection{Related Work}

Perhaps the most direct approach for analyzing the RICs of the KR product matrix
is to use the eigenvalue interlacing theorem \cite{HornAndJohnson}, which 
relates the singular values of any $k$-column submatrix of the KR product between two matrices to the singular values of their Kronecker product. This is possible because any $k$ 
columns of the KR product can together be interpreted as a submatrix of the Kronecker product. 
However, barring the maximum and minimum singular values of the Kronecker product, there is no available explicit characterization of its non-extremal singular values that can be used to obtain tight bounds 
for the $k$-RIC of the KR product. Bounding the RIC using the extreme singular values of the Kronecker product matrix turns out to be too loose to be useful. In this context, it is noteworthy to mention that an upper bound for the $k$-RIC of the Kronecker product is derived in terms of the $k$-RICs of the input matrices in 
\cite{Duarte12KroneckerCS, SadeghJokar10KCS}. However, the $k$-RIC of the Khatri-Rao product is yet to be analyzed.

Recently, \cite{BhaskaraMoitra14SmoothedAna, AndersonBelinGoyal14} gave probabilistic 
lower bounds for the minimum singular value of the columnwise KR product between two or 
more matrices. These bounds are limited to randomly constructed input matrices, and are 
polynomial in the matrix size. 
%
In \cite{Sidiropoulos2000Kruskal}, it is shown that for any two matrices $\matA$ and $\matB$, 
the Kruskal-rank\footnote{The Kruskal rank or K-rank of any matrix $\matA$ is the largest integer $r$ such that any $r$ columns of $\matA$ are linearly independent.} of $\AkrB$ has a lower bound in terms of K-rank($\matA$) and K-rank$(\matB)$. 
In fact, K-rank($\AkrB$) is at least as high as $\max{\lb \text{K-rank}(\matA), \text{K-rank}(\matB) \rb}$, 
thereby suggesting that $\AkrB$ exhibits a stronger restricted isometry property than both $\matA$ and $\matB$. 

A closely related yet weaker notion of restricted isometry constant is the $\tau$-robust K-rank, 
denoted by $\text{K-rank}_{\tau}$. For a given matrix $\matPhi$, the $\text{K-rank}_{\tau}(\matPhi)$ is defined as the largest $k$ for which every $n \times k$ submatrix of $\matPhi$ has its smallest 
singular value larger than $1/\tau$. In \cite{BhaskaraMoitra14SmoothedAna}, it is 
shown that the $\tau$-robust K-rank is super-additive, implying that the $\text{K-rank}_{\tau}$ 
of the Khatri-Rao product is strictly larger than individual $\text{K-rank}_{\tau}$s of the  
input matrices. 

In \cite{Dasarathy15CovMatEst}, the isometry property of the Kronecker product $\matA \otimes \matB$ is analyzed in the $\ell_{1}$-norm sense for the restricted class of input vectors expressible as vectorized $d$-distributed sparse matrices, wherein, $\matA$ and $\matB$ are the adjacency matrices of two independent uniformly random $\delta$-left regular bipartite graphs. In this paper, we instead analyze the restricted isometry of the columnwise Khatri-Rao product $\AkrB$, which is equivalent to the RIP of the Kronecker product $\matA \otimes \matB$ with respect to vectorized sparse 
diagonal matrices. In our work, we assume that input matrices $\matA$ and $\matB$ are random 
with independent subgaussian elements.

\subsection{Main Contributions}
We derive two kinds of upper bound on the $k$-RIC 
of the columnwise Khatri-Rao product of $m \times n$ sized matrices $\matA$ and $\matB$. 
The bounds are briefly described below.
\begin{enumerate}
	\item[1)] A deterministic upper bound for the $k$-RIC of $\matA \odot \matB$ in terms of the 
	      $k$-RICs of the input matrices $\matA$ and $\matB$. The bound is valid 
	      for $k \le m$, and for input matrices with unit $\ell_{2}$-norm columns.
	\item[2a)] A probabilistic upper bound for the $k$-RIC of $\matA \odot \matB$ in terms of 
	      $m, n$ and $k$, for real valued random matrices $\matA, \matB$ containing 
	      i.i.d. subgaussian elements. The probabilistic bound is polynomially tight with respect 
	      to the input matrix dimension $n$. 
	\item[2b)] A probabilistic upper bound for the $k$-RIC of the self KR product 
	$\matA \odot \matA$ in terms of $m, n$, and $k$, where $\matA$ is a real valued matrix containing i.i.d. subgaussian elements. 
\end{enumerate}

A key idea used in our RIC analysis is the fact (stated formally as Proposition \ref{prop_gramian_kr_as_hadamard}) that given any two matrices $\matA$ and $\matB$, the Gram matrix of their Khatri-Rao product $(\matA \odot \matB)^{H}(\matA \odot \matB)$ can be interpreted as the Hadamard product (element wise multiplication) between Gram matrices $\matA^{H} \matA$ and $\matB^{H} \matB$. The Hadamard product form turns out to be more analytically tractable than columnwise Kronecker product form of the KR matrix. 

Lately, in several machine learning problems,  
the necessary and sufficient conditions for successful signal recovery have been reported in terms 
of the RICs of a certain Khatri-Rao product matrix serving as \textcolor{black}{a \emph{pseudo}} sensing matrix~\cite{Ma10DoaEstViaKR, PPal15CovMMV}. 
In light of this, our proposed RIC bounds are quite timely, and pave the way towards obtaining order-wise 
sample complexity bounds in several fundamental learning problems.

The rest of this article is organized as follows.
In Secs. \ref{sec:kr_k_ric_bound_deterministic} and \ref{sec:kr_k_ric_bound_probabilistic}, 
we present our main results: deterministic and probabilistic RIC bounds, respectively, for 
a generic columnwise KR product matrix. 
Sec. \ref{sec:kr_k_ric_bound_probabilistic} also discusses about the RIP of the 
\emph{self} Khatri-Rao product of a matrix with itself, an important matrix type encountered in the sparse diagonal covariance matrix estimation 
problem. 
In Secs. \ref{sec:bkgd_deterministic_bnd} and \ref{sec:bkgd_probabilistic_bnd}, we present some 
background concepts required for proving the proposed RIC bounds. 
Secs. \ref{sec:proof_deterministic} and \ref{sec:proof_probabilistic} provide the detailed proofs of the 
deterministic and probabilistic RIC bounds, respectively.
Final conclusions are presented in Sec. \ref{sec:concluding_remarks}.

\emph{Notation:} In this work, bold lowercase letters are used for representing both scalar random variables 
as well as vectors. 
Bold uppercase letters are reserved for matrices. The $\ell_2$-norm of a vector $\vecx$ is denoted by $\lvv \vecx \rvv_{2}$.
For an $m \times n$ matrix $\matA$, $\lvv \matA \rvv$ denotes its operator norm, 
$\lvv \matA \rvv \triangleq \sup_{\vecx \in \Real^{n}, \vecx \neq 0} \frac{\lvv \matA \vecx \rvv_{2}}{\lvv \vecx \rvv_{2}} $. 
The Hilbert-Schmidt (or Frobenius) norm of $\matA$ is defined as $\lvv \matA \rvv_{HS} = \sum_{i = 1}^{m} \sum_{j = 1}^{n} \lv A_{i, j} \rv^{2}$. The symbol $[n]$ denotes the index set $\lc 1, 2, \dots, n \rc$. 
For any index set $\setS \subseteq [n]$, $\matA_{\setS}$ denotes the submatrix comprising the columns 
of $\matA$ indexed by $\setS$. 
The matrices 
$\matA \otimes \matB$, $\matA \circ \matB$ and $\matA \odot \matB$ denote the Kronecker product, 
Hadamard product and columnwise Khatri-Rao product of $\matA$ and $\matB$, respectively. 
$\matA \le \matB$ implies that $\matB - \matA$ is a positive semidefinite matrix. $\matA^{T}$, $\matA^{H}$, $\matA^{-1}$, and $\matA^{\dagger}$ denote the 
transpose, conjugate-transpose, inverse and generalized matrix inverse operations of $\matA$, respectively. 
$\mathbb{E}(\vecx)$ denotes the expectation of the random variable $\vecx$. 
$\mathbb{P}(\mathcal{E})$ denotes the probability of event $\mathcal{E}$. $\mathbb{P}(\mathcal{E} \vert \mathcal{E}^{\prime})$ denotes the conditional probability of event $\mathcal{E}$ given event $\mathcal{E}^{\prime}$ has occurred.

\section{Deterministic $k$-RIC Bound} \label{sec:kr_k_ric_bound_deterministic}

In this section, we present our first upper bound on the $k$-RIC of a generic columnwise KR product $\AkrB$, for any 
two same-sized matrices $\matA$ and $\matB$ with normalized columns. The bound is given in terms of the 
$k$-RICs of $\matA$ and $\matB$.

\begin{theorem} \label{thm_ric_bound_for_kr}
 Let $\matA$ and $\matB$ be $m \times n$ sized real-valued matrices with unit $\ell_{2}$-norm columns 
 and satisfying the $k^{\text{th}}$ order restricted isometry property with constants $\delta^{\matA}_{k}$ and $\delta^{\matB}_{k}$, respectively.
 Then, their columnwise Khatri-Rao product $\AkrB$ satisfies the restricted isometry property with 
 $k$-RIC at most $\delta^{2}$, where $\delta \triangleq 
 \text{max} \lb \delta^{\matA}_{k}, \delta^{\matB}_{k} \rb$, i.e.,
 \begin{equation} \label{rip_of_kr}
  (1 - \delta^{2}) ||\vecz||_{2}^{2}
  \le
  || (\AkrB) \vecz||_{2}^{2}
  \le
  (1 + \delta^{2}) ||\vecz||_{2}^{2}
 \end{equation}
 holds for all $k$-sparse vectors $\vecz \in \Real^{n}$.
\end{theorem}
\begin{proof}
The proof is provided in Section \ref{sec:proof_deterministic}.
\end{proof}

\textit{Remark 1:} The RIC bound for $\AkrB$ in Theorem~\ref{thm_ric_bound_for_kr} is relevant only when 
$\delta_{k}(\matA)$ and $\delta_{k}(\matB)$ lie in $\lb 0, 1 \rb$, which is true only for $k \le m$. 
In other words, the above $k$-RIC characterization for $\AkrB$ requires the input matrices $\matA$ and $\matB$ 
to be $k$-RIP compliant.

\textit{Remark 2:} Since the input matrices $\matA$ and $\matB$ satisfy $k$-RIP with 
$\delta_{k}(\matA), \delta_{k}(\matB) \in (0, 1)$, it follows from Theorem~\ref{thm_ric_bound_for_kr} that 
$\delta_{k}(\AkrB)$ is strictly smaller than $\max \lb \delta_{k}(\matA), \delta_{k}(\matB) \rb$. 
If $\matB = \matA$, the special case of self Khatri-Rao product $\matA \odot \matA$ arises, for which 
\begin{equation} \label{self_kr_ric_bnd}
\delta_{k}(\matA \odot \matA) \;\textcolor{black}{<} \;  \delta_{k}^{2}(\matA).   
\end{equation}
Above implies that the self Khatri-Rao product $\matA \odot \matA$ is a better restricted isometry 
compared to $\matA$ itself. This observation is in alignment with the 
expanding Kruskal rank and shrinking mutual coherence of the self Khatri-Rao product reported in \cite{PPal15CovMMV}.
In fact, for $k = 2$, the $2$-RIC bound \eqref{self_kr_ric_bnd} exactly matches the mutual coherence bound shown in~\cite{PPal15CovMMV}.



For $k \in (m, m^2]$, using (\cite{KoiranZouzias14RIPCert}, Theorem 1), one can show that $\delta_{k}(\AkrB) \le \lb \sqrt{k} + 1 \rb \delta_{\sqrt{k}} $, 
where $\delta_{\sqrt{k}} = \text{max}\lb\delta_{\sqrt{k}}(\matA), \delta_{\sqrt{k}}(\matB) \rb$. This bound, however, 
loses its tightness and quickly becomes unattractive for larger values of $k$. Finding a tighter $k$-RIC upper bound 
for the $k > m$ case remains an open problem.

To gauge the tightness of the proposed $k$-RIC bound for $\AkrB$, we present its 
simulation-based quantification for the case when the input matrices $\matA$ and $\matB$ 
are random Gaussian matrices with i.i.d. $\mathcal{N}(0, 1/m)$ entries. Fig.~\ref{fig_kr_ric_char} plots 
$\delta_{k}(\matA)$, $\delta_{k}(\matB)$, $\delta_{k}(\AkrB)$ and the upper bound  
$\overline{\delta_{k}(\AkrB)} = \lb \max{ \lb \delta_{k}(\matA), \delta_{k}(\matB) \rb } \rb^{2}$ for a 
range of input matrix dimension $m$. 
The aspect ratio $m/n$ of the input matrices is fixed to $0.5$.\footnote{
\textcolor{black}{
While the $m \times n$ matrices $\matA$ and $\matB$ may represent 
highly underdetermined linear systems (when $m \ll n$), their $m^2 \times n$ 
sized Khatri-Rao product $\AkrB$ can become an overdetermined system. 
In fact, many covariance matching based sparse support recovery algorithms \cite{PPal15CovMMV,WipfRao07MSBL, Khanna17RDCMP} 
exploit this fact to offer significantly better support reconstruction performance.}} 
For computational tractability, we restrict our analysis to the cases $k = 2$ and $3$.
The RICs: $\delta_{k}(\matA)$, $\delta_{k}(\matB)$ and 
$\delta_{k}(\AkrB)$ are computed by exhaustively searching for the worst conditioned 
submatrix comprising $k$ columns of $\matA$, $\matB$ and $\AkrB$, respectively. 
From Fig.~\ref{fig_kr_ric_char}, we observe that the proposed $k$-RIC upper bound becomes tighter 
as the input matrices grow in size.  
Interestingly, the experiments suggest that the KR product $\AkrB$ satisfies 
$k$-RIP in spite of the input matrices $\matA$ and $\matB$ failing to do so. A theoretical confirmation 
of this empirical observation is an interesting open problem. 
\begin{figure}[t]
\centering
\includegraphics[width=0.84\columnwidth]{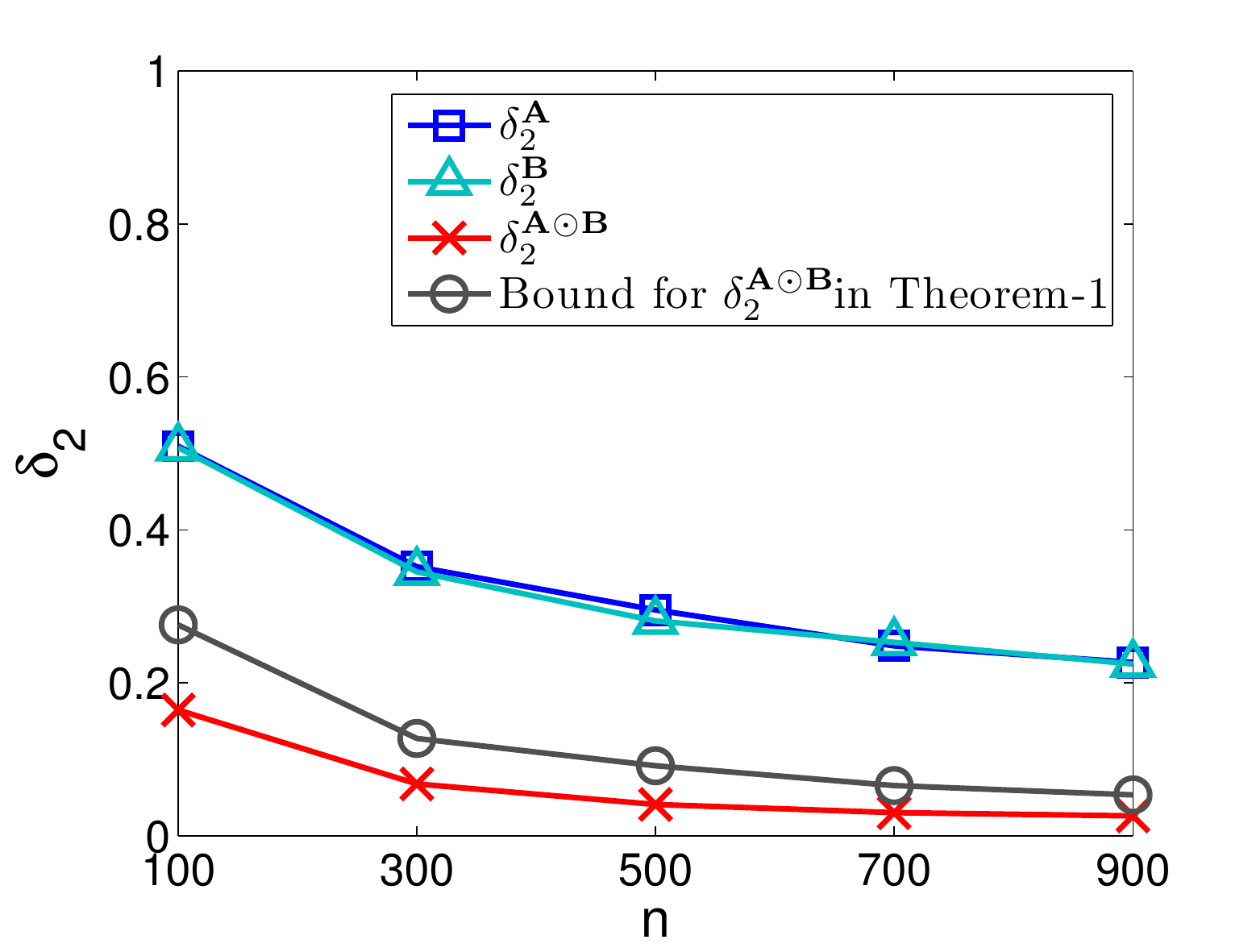}
\includegraphics[width=0.84\columnwidth]{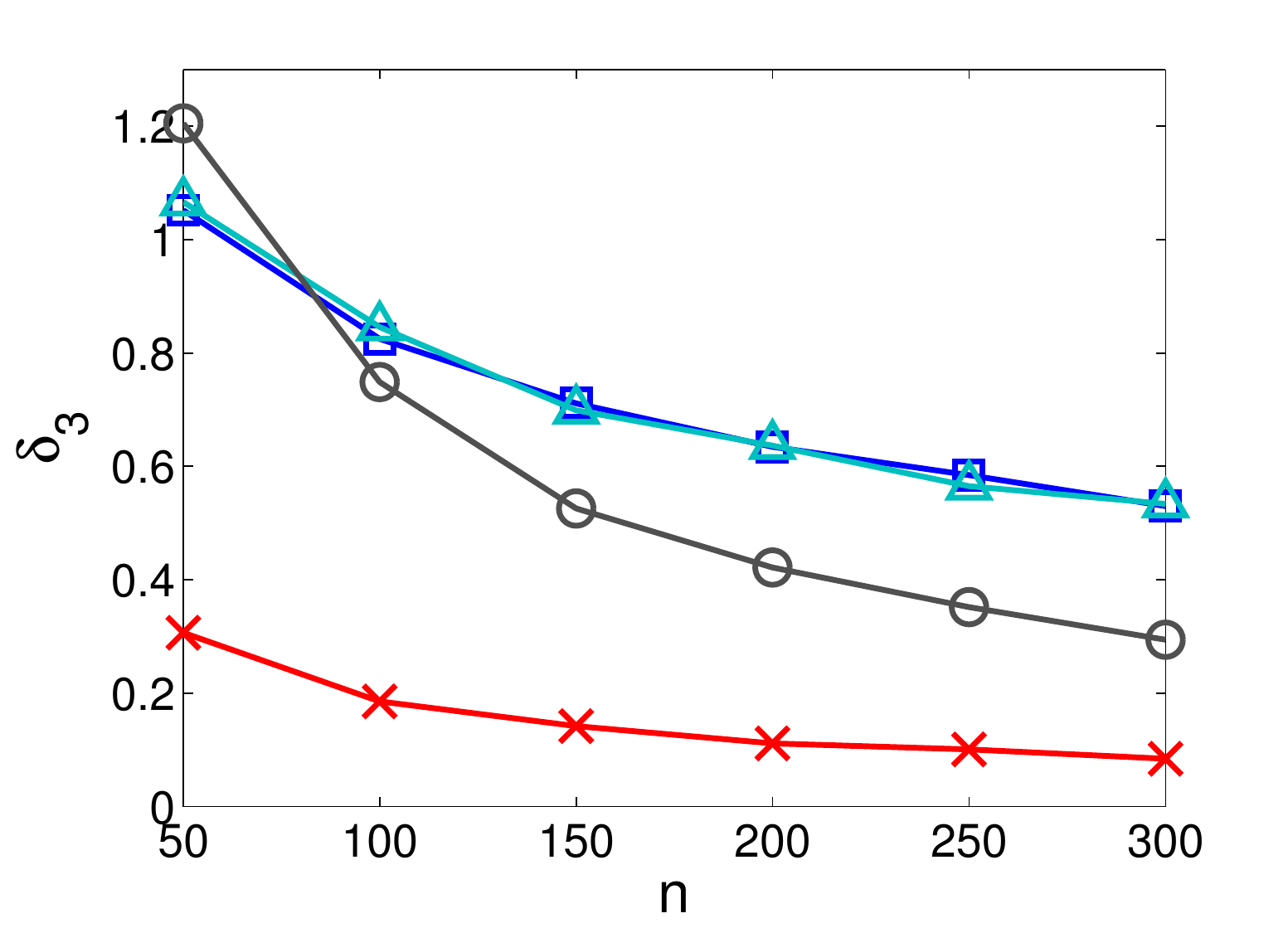}
\caption{Variation of $k$-RICs of $\matA$, $\matB$, $\AkrB$ and the proposed upper bound with increasing input matrix dimensions.
The top and the bottom plots are for $k = 2$ and $3$, respectively. Each data point is averaged over $10$ trials.}
\label{fig_kr_ric_char}
\end{figure}

\section{Probabilistic $k$-RIC Bound} \label{sec:kr_k_ric_bound_probabilistic}
The deterministic RIC bound for the columnwise KR product discussed in the Section~\ref{sec:kr_k_ric_bound_deterministic} 
is useful only when the input matrices have unit norm columns. 
While the use of column normalized sensing matrices is commonplace in compressive sensing, we are often interested in showing the restricted isometry of the Khatri-Rao product of randomly constructed input matrices with columns being normalized only in the average sense. 
This concern is addressed by our second RIC bound which is of probabilistic type and is applicable to the columnwise KR product of random input
matrices containing i.i.d. subgaussian entries. Below, we define a subgaussian random variable 
and state some of its properties.

\begin{defn}{(Subgaussian Random Variable)}: A zero mean random variable $\vecx$ is called subgaussian, if its tail probability is dominated by that of a Gaussian random variable. 
In other words, there exist constants~$C, K > 0$ such that 
 $\PP(|\vecx| \ge t) \le Ce^{-t^2/K^2}$ for $t > 0$.
\end{defn}
Gaussian, Bernoulli and all bounded random variables are subgaussian 
random variables. 
For a subgaussian random variable, its $p^{\text{th}}$ order moment grows no faster than 
$O(p^{p/2})$ \cite{Vershynin10RMT}. In other words, 
there exists $K_{1} > 0$ such that
\begin{equation}
\lb \EE \lv \vecx \rv^{p} \rb^{\frac{1}{p}} \le K_{1} \sqrt{p}, \;\; p \ge 1. 
\end{equation}
The minimum such $K_{1}$ is called the \textit{subgaussian or $\psi_{2}$ norm} of the random variable $\vecx$, i.e.,
\begin{equation}
 \lvv \vecx \rvv_{\psi_{2}} = \sup_{p \ge 1} p^{-1/2} \lb \EE \lv \vecx \rv^{p} \rb^{\frac{1}{p}}.
\end{equation}

Given a pair of random input matrices with i.i.d.\ subgaussian entries, 
Theorem~\ref{thm_kric_bound_X_neq_Y} presents an upper bound on the $k$-RIC of their columnwise KR product.
\begin{theorem} \label{thm_kric_bound_X_neq_Y}
 Suppose $\matA$ and $\matB$ are $m \times n$ random matrices with real i.i.d.\ subgaussian entries, 
 such that $\mathbb{E}\matA_{ij} = 0$, $\mathbb{E}\matA_{ij}^{2} = 1$, and 
 $\lvv \matA_{ij}\rvv_{\psi_{2}} \le \kappa$, and similarly for $\matB$. 
Then, the $k^{\text{th}}$ order restricted isometry constant of $\frac{\matA}{\sqrt{m}} \odot \frac{\matB}{\sqrt{m}}$, 
denoted by $\delta_{k}$, satisfies $\delta_{k} \le \delta$ with probability at least 
$1 - 10 n^{-2(\gamma-1)}$ for any $\gamma > 1$, provided that
\begin{equation}
 m \ge 4 c \gamma \kappa_{o}^{4} \lb  \frac{k \log{n}}{\delta} \rb,
 \nonumber 
\end{equation}
where $\kappa_{o} = \max \lb \kappa, 1 \rb$ and $c$ a universal positive constant.
\end{theorem}
\begin{proof}
The proof is provided in Section \ref{sec:proof_probabilistic}.
\end{proof}

The normalization constant $\sqrt{m}$ used while computing the KR product 
$\frac{\matA}{\sqrt{m}} \odot \frac{\matB}{\sqrt{m}}$ 
ensures that the columns of the input matrices 
$\frac{\matA}{\sqrt{m}}$, $\frac{\matB}{\sqrt{m}}$ have unit average energy, i.e. 
$\EE \lvv \veca_{i}/\sqrt{m} \rvv_{2}^{2} = \EE \lvv \vecb_{i} / \sqrt{m} \rvv_{2}^{2} = 1 \text{ for } 1 \le i \le n$. 
Column normalization is a key assumption towards correct modelling of the isotropic, 
norm-preserving nature of the effective sensing matrix $\frac{1}{m}(\AkrB)$, 
an attribute desired in most sensing matrices employed in practice. 

Theorem~\ref{thm_kric_bound_X_neq_Y} implies that 
\begin{equation} \label{kric_ub_high_prob}
 \delta_{k}\lb \frac{\matA}{\sqrt{m}} \odot \frac{\matB}{\sqrt{m}} \rb \le
 O \lb \frac{ k \log{n}}{m} \rb
\end{equation}
with arbitrarily high probability.
Thus, the above $k$-RIC bound decreases as $m$ increases, which is intuitively appealing. 
Interestingly, for fixed $k$ and $n$, the above $k$-RIC upper bound 
for $\frac{\matA}{\sqrt{m}} \odot \frac{\matB}{\sqrt{m}}$ decays as $O(\frac{1}{m})$. This is a significant improvement over the 
$O(\frac{1}{\sqrt{m}})$ decay rate \cite{FoucartRauhut13CSBook} already known for the individual $k$-RICs of the input subgaussian matrices $\frac{\matA}{\sqrt{m}}$ and $\frac{\matB}{\sqrt{m}}$. 
Thus, for any $m$, the Khatri-Rao product exhibits stronger restricted isometry property, with significantly smaller $k$-RICs compared to the $k$-RICs for the input matrices.

%

In some applications, the effective sensing matrix is expressible as the self-Khatri Rao product $\matX \odot \matX$ of a certain system matrix $\matX$ with itself~\cite{PPal15CovMMV}. In Theorem~\ref{thm_kric_bound_X_eq_Y} below, we present a $k$-RIC bound for  self-Khatri-Rao product matrices.
\begin{theorem} \label{thm_kric_bound_X_eq_Y}
 Let $\matA$ be an $m \times n$ random matrix with real i.i.d.\ subgaussian entries, 
 such that $\mathbb{E}\matA_{ij} = 0$, $\mathbb{E}\matA_{ij}^{2} = 1$, and 
 $\lvv \matA_{ij}\rvv_{\psi_{2}} \le \kappa$. 
 Then, the $k^{\text{th}}$ order restricted isometry constant 
 of the column normalized self Khatri-Rao product $\frac{\matA}{\sqrt{m}} \odot \frac{\matA}{\sqrt{m}}$ 
 satisfies $\delta_{k} \le \delta$ with probability at least $1 - 5 n^{-2(\gamma-1)}$ for any $\gamma \ge 1$, provided 
\begin{equation} 
 m \ge 4 c^{\prime} \gamma \kappa_{o}^{4}  \lb \frac{k \log{n}}{\delta} \rb.
 \nonumber 
\end{equation}
Here, $\kappa_{o} = \max \lb \kappa, 1 \rb$ and $c^{\prime}>0$ is a universal constant.
\end{theorem}
\begin{proof}
See Appendix \ref{App:proof_thm_kric_bound_X_eq_Y}.
\end{proof}
Theorem \ref{thm_kric_bound_X_eq_Y} implies that 
\begin{equation} \label{eq:ric_bound_aodota}
 \delta_{k}\lb \frac{\matA}{\sqrt{m}} \odot \frac{\matA}{\sqrt{m}} \rb \le O \lb \frac{ k \log{n}} {m} \rb
\end{equation}
with arbitrarily high probability. The above $k$-RIC bound for the self Khatri-Rao product scales with $m, n$, and $k$ in a similar fashion as the asymmetric Khatri-Rao product. 

\textit{Remark 3:} From Theorem~\ref{thm_kric_bound_X_eq_Y}, the $k$-RIC of the columnwise Khatri-Rao product of two $m \times n$ sized subgaussian matrices can be guaranteed to be less than unity for $m = O\lb k \log n \rb$. We conjecture that the Khatri-Rao product continues to be $k$-RIP compliant even when $m < k$. However, proving such a result is beyond the scope of this work.

\section{Preliminaries for Deterministic $k$-RIC Bound} \label{sec:bkgd_deterministic_bnd}
In this section, we present some preliminary concepts 
and results which are necessary for the derivation of the deterministic $k$-RIC bound in Theorem~\ref{thm_ric_bound_for_kr}.
For the sake of brevity, we provide proofs only for claims that have not 
been explicitly shown in their cited sources. 

\subsection{Properties of the Kronecker and Hadamard product}
For any two matrices $\matA$ and $\matB$ of dimensions $m \times n$ and $p \times q$, the kronecker product $\matA \otimes \matB$ is 
the $mp \times nq$ matrix
\begin{equation} \label{defn_kron_prod}
 \matA \otimes \matB = 
\begin{pmatrix}
a_{11} \matB & a_{12} \matB & \cdots & a_{1n} \matB \\
a_{21} \matB & a_{22} \matB & \cdots & a_{2n} \matB \\
\vdots  & \vdots  & \ddots & \vdots  \\
a_{m1} \matB & a_{m2} \matB & \cdots & a_{mn} \matB 
\end{pmatrix}.
\end{equation}

The following Proposition relates the spectral properties of the Kronecker product and its constituent matrices.
\begin{proposition}[$7.1.10$ in \cite{Bernstein09MatrixMath}] \label{kron_svd}
Let $\matA \in \Real^{n \times n}$ and $\matB \in \Real^{p \times p}$ admit eigenvalue decompositions 
$\matU_{A}\Lambda_{A}\matU_{A}^{T}$ and $\matU_{B}\Lambda_{B}\matU_{B}^{T}$, respectively. Then, 
\begin{equation}
(\matU_{A} \otimes \matU_{B}) (\Lambda_{A} \otimes \Lambda_{B}) (\matU_{A} \otimes \matU_{B})^{T}
\nonumber 
\end{equation}
yields the eigenvalue decomposition for $\AkB$.
\end{proposition}
For any two matrices of matching dimensions, say $m \times n$, their Hadamard product $\matA \circ \matB$ is obtained by elementwise multiplication of the entries of 
the input matrices, i.e., 
\begin{equation}
 [\matA \circ \matB]_{i, j} = a_{ij} b_{ij}  \hspace{0.2cm} \text{ for } i \in [m], j \in [n]. 
\end{equation}
The Hadamard product $\matA \circ \matB$ is a principal submatrix of the Kronecker product
$\matA \otimes \matB$ \cite{LiuTrenkler, Visick2000HadamardAsSubMat}. For $n \times n$ sized square matrices $\matA$ and $\matB$, one can write,
\begin{equation} \label{hadamard_as_a_kron}
\AhB = \matJ^{T} (\AkB) \matJ,
\end{equation}
where $\matJ$ is an $n^2 \times n$ sized selection matrix constructed entirely from $0$'s and $1$'s which satisfies $\matJ^{T}\matJ = \matI_{n}$. 

In Proposition~\ref{prop_emax_rman_cmax}, we present an upper bound on the spectral radius of a generic Hadamard product.
\begin{proposition} \label{prop_emax_rman_cmax}
For every $\matA, \matB \in \Real^{m \times n}$, we have
\begin{equation} \label{emax_rman_cmax_ineq}
\sigma_{\text{max}} \lb \AhB \rb \le r_{\text{max}}(\matA) c_{\text{max}}(\matB)  
\end{equation}
where $\sigma_{\text{max}}(\cdot)$, $r_{\text{max}}(\cdot)$ and $c_{\text{max}}(\cdot)$ are the largest singular value, the largest row $\ell_{2}$-norm and the largest 
column $\ell_{2}$-norm of the input matrix, respectively. 
\end{proposition}
\begin{proof}
See Theorem $5.5.3$ in \cite{horn94}. 
\end{proof}
We now state an important result about the Hadamard product of two positive semidefinite matrices. 
\begin{proposition} [Mond and Pe\u{c}ari\'{c} \cite{MondPecaric1998}] \label{prop_hadamard_prod_pow_ineq}
Let $\matA$ and $\matB$ be positive semidefinite $n \times n$ Hermitian matrices and let $r$ and $s$ be two 
nonzero integers such that $s > r$. Then,
\begin{equation}
\lb \matA^{s} \circ \matB^{s} \rb^{1/s} \ge \lb \matA^{r} \circ \matB^{r} \rb^{1/r}.   
\end{equation}
\end{proposition}

In Propositions \ref{corr_mat_sqrt_hadamard_is_double_stochastic} and \ref{prop_hadamard_corr_matrix_less_than_I}, we state some spectral 
properties of correlation matrices and their Hadamard products. Correlation matrices are symmetric positive semidefinite matrices 
with diagonal entries equal to one. Later on, we will exploit the fact that the singular values of the columnwise KR product are related 
to the singular values of the Hadamard product of certain correlation matrices.

\begin{proposition} \label{corr_mat_sqrt_hadamard_is_double_stochastic}
If $\matA$ is an $n \times n$ correlation matrix, then $\matA^{1/2} \circ \matA^{1/2}$ is a doubly stochastic matrix.
\end{proposition}
\begin{proof}
See Appendix \ref{App:proof_corr_mat_sqrt_hadamard_is_double_stochastic}.
\end{proof}

\begin{proposition} [Werner \cite{Image26}] \label{prop_hadamard_corr_matrix_less_than_I}
For any correlation matrices $\matA$ and $\matB$ of the same size, we have $\matA^{1/2} \circ \matB^{1/2} \le \matI$,
where $\matA^{1/2}$ and $\matB^{1/2}$ are the positive square roots of $\matA$ and $\matB$, respectively. 
\end{proposition}
\begin{proof}
Since $\matA$ is a correlation matrix, from Proposition \ref{corr_mat_sqrt_hadamard_is_double_stochastic}, it follows that
$\matA^{1/2} \circ \matA^{1/2}$ is doubly stochastic. Since the rows and columns of 
$\matA^{1/2} \circ \matA^{1/2}$ sum to unity, we have $r_{\text{max}}(\matA^{1/2}) = c_{\text{max}}(\matA^{1/2}) = 1$.
Similarly, $r_{\text{max}}(\matB^{1/2}) = c_{\text{max}}(\matB^{1/2}) = 1$. Then, from Proposition \ref{prop_emax_rman_cmax}, it follows 
that the largest eigenvalue of $\matA^{1/2} \circ \matB^{1/2}$ is at most unity.
\end{proof}

\subsection{Matrix Kantorovich Inequalities}
Matrix Kantorovich inequalities relate positive definite matrices 
by inequalities in the sense of the L\"{o}wner partial order.\footnote{The L\"{o}wner partial order here refers to the relation 
``$\le$". For positive definite matrices $\matA$ and $\matB$, $\matA \le \matB$ {if and only if} $\matB - \matA$ is a positive 
{semi-}definite matrix.} These inequalities can be used to extend the 
L\"{o}wner partial order to the Hadamard product of positive definite matrices.  
Our proposed RIC bound relies on the tightness of these Kantorovich inequalities and their extensions.

A matrix version of the Kantorovich inequality was first proposed by Marshall and Olkin in \cite{MarshallOlkin64}. It is 
stated below as Proposition~\ref{prop_narshall_n_olkin}.

\begin{proposition} [Marshall and Olkin \cite{MarshallOlkin64}] \label{prop_narshall_n_olkin}
Let $\matA$ be an $n \times n$ positive definite matrix. Let $\matA$ admit the Schur decomposition $\matA = \matU \matLambda \matU^{T}$ with unitary $\matU$ and 
$\matLambda = \diag{ \lb \lambda_{1}, \lambda_{2}, \dots, \lambda_{n} \rb}$ such that $\lambda_{i} \in [m, M]$. Then, we have
\begin{equation} \label{matrix_kantorovich_ineq}
\matA^{2} \le (M + m) \matA - m M \matI_{n}.
\end{equation}
\end{proposition}

The above inequality \eqref{matrix_kantorovich_ineq} is the starting point for obtaining a variety of forward and reverse 
Kantorovich-type matrix inequalities for positive definite matrices. In Propositions~\ref{prop_matrix_kantorovitch_ineq_form2} 
and~\ref{prop_liu_neudecker_kantorovitch_with_inverse}, we state specific forward and reverse inequalities, respectively, 
which are relevant to us.
\begin{proposition} [Liu and Neudecker \cite{LiuNeudecker96}] \label{prop_matrix_kantorovitch_ineq_form2}
Let $\matA$ be an $n \times n$ positive definite Hermitian matrix, with eigenvalues in $[m, M]$. Let $\matV$ be an $n \times n$ 
matrix such that $\matV^{T} \matV = \matI_{n}$. Then, 
\begin{equation} \label{matrix_kantorovich_ineq_2}
\matV^{T} \matA^{2} \matV  - \lb \matV^{T} \matA \matV \rb^{2} \le \frac{1}{4} (M - m)^{2} \matI_{n}.
\end{equation}
\end{proposition}
\begin{proposition} [Liu and Neudecker \cite{LiuNeudecker97}] \label{prop_liu_neudecker_kantorovitch_with_inverse}
Let $\matA$ and $\matB$ be $n \times n$ positive definite matrices. Let $m$ and $M$ be the minimum and maximum eigenvalues of 
$\matB^{1/2} \matA^{-1} \matB^{1/2}$. Let $\matX$ be an $n \times p$ matrix.
Then, we have
\begin{equation} \label{matrix_kantorovich_ineq_4}
(\matX^{T} \matB \matX) (\matX^{T} \matA \matX)^{\dag} (\matX^{T} \matB \matX)
\ge 
\frac{4 m M}{(M + m)^{2}} 
 \matX^{T} \matB \matA^{-1} \matB \matX.
\end{equation}
\end{proposition}
Proposition~\ref{prop_matrix_kantorovitch_ineq_form2} can be proved using \eqref{matrix_kantorovich_ineq} by pre- and post-multiplying by 
$\matV^{T}$ and $\matV$, respectively, followed by completion of squares for the right hand side terms. 
The proof of Proposition~\ref{prop_liu_neudecker_kantorovitch_with_inverse} is given in \cite{LiuNeudecker97}.


\subsection{Kantorovich Matrix Inequalities for the Hadamard Products of Positive Definite Matrices}\label{sec:kantorovich_for_hadamard}
Lemmas~\ref{lem_matrix_kantorovich_ineq_2_hadamard_ver} and \ref{lem_reverse_kantorovich_hadamard_ver} stated below
extend the Kantorovich inequalities from the previous subsection to Hadamard products. 

\begin{lemma} [Liu and Neudecker \cite{LiuNeudecker96}] \label{lem_matrix_kantorovich_ineq_2_hadamard_ver}
Let $\matA$ and $\matB$ be $n \times n$ positive definite matrices, with $m$ and $M$ denoting the minimum and maximum 
eigenvalues of $\AkB$. Then, we have
\begin{equation} 
\matA^{2} \circ \matB^{2} \le (\matA \circ \matB)^{2} + \frac{1}{4}(M - m)^{2} \matI_{n}. 
\end{equation}
\end{lemma}
\begin{proof}
Let $\matJ$ be the selection matrix such that $\matJ^{T}\matJ = \matI_{n}$ and $\AhB = \matJ^{T} (\AkB) \matJ$.
Then, by applying Proposition \ref{prop_matrix_kantorovitch_ineq_form2} with $\matA$ replaced with $\AkB$, and $\matV$ replaced with $\matJ$, we obtain
\begin{equation}
 \matJ^{T}(\AkB)^{2}\matJ - (\matJ^{T}(\AkB)\matJ)^{2} \le \frac{1}{4}(M-m)^{2}\matI_{n}. 
 \nonumber
\end{equation}
Using Fact~$8.21.29$ in \cite{Bernstein09MatrixMath}, i.e., $(\AkB)^{2} = \matA^{2} \otimes \matB^{2}$, we get
\begin{equation}
  \matJ^{T}(\matA^{2} \otimes \matB^{2})\matJ - (\AhB)^{2} \le \frac{1}{4}(M-m)^{2}\matI_{n},
  \nonumber 
\end{equation}
Finally, by observing that $\matJ^{T} (\matA^{2} \otimes \matB^{2}) \matJ = \matA^{2} \circ\matB^{2}$, we obtain the desired result.
\end{proof}

\begin{lemma} [Liu \cite{LiuS2002EconometricTheory}] \label{lem_reverse_kantorovich_hadamard_ver}
Let $\matA$, $\matB$ be $n \times n$ positive definite correlation matrices. Then,
\begin{equation} 
\matA^{1/2} \circ \matB^{1/2} \ge \frac{2 \sqrt{mM}}{ m + M} \matI_{n}
\end{equation}
where the eigenvalues of $\matA$ and $\matB$ lie inside $[m, M]$.
\end{lemma}
Lemma~\ref{lem_reverse_kantorovich_hadamard_ver} follows from Proposition \ref{prop_liu_neudecker_kantorovitch_with_inverse}, by 
replacing $\matA$ with $\matI_{n} \otimes \matB$, $\matB$ with $\matA^{1/2} \otimes \matB^{1/2}$, and $\matX$ with $\matJ$, 
where $\matJ$ is the $n^{2} \times n$ binary selection matrix such that $\matJ^{T}\matJ = \matI$ and $\AhB = \matJ^{T} (\AkB) \matJ$.

\section{Proof of the Deterministic $k$-RIC Bound (Theorem~\ref{thm_ric_bound_for_kr})} \label{sec:proof_deterministic} 
The key idea used in bounding the $k$-RIC of the columnwise KR product $\matA \odot \matB$ 
is the observation that the Gram matrix of $\AkrB$ can 
be interpreted as a Hadamard product between the two correlation matrices 
$\matA^{T}\matA$ and $\matB^{T}\matB$, as mentioned in the following Proposition.
\begin{proposition}[Rao and Rao \cite{RaoNRao}] \label{prop_gramian_kr_as_hadamard}
For $\matA, \matB \in \Real^{m \times n}$, 
\begin{equation} \label{kr_as_hadamard}
(\AkrB)^{T}(\AkrB) = (\matA^{T}\matA) \circ (\matB^{T}\matB)
\end{equation}
\end{proposition}
\begin{proof}
See Proposition $6.4.2$ in \cite{RaoNRao}. 
\end{proof}

Then, by using the forward and reverse Kantorovich matrix inequalities,
we obtain the proposed upper bound for $k$-RIC of $\AkrB$ in Theorem~\ref{thm_ric_bound_for_kr} 
as explained in the following arguments. 

Without loss of generality, let $S \subset [n]$ be an arbitrary index set representing the nonzero 
support of $\vecz$ in \eqref{rip_of_kr}, with $|S| \le k$. Let $\matA_{S}$ denote the $m \times |S|$  submatrix of $\matA$, constituting $|S|$ 
columns of $\matA$ indexed by the set $S$. Let $\matB_{S}$ be constructed similarly.
Since $\delta_{k}(\matA)$, $\delta_{k}(\matB) < 1$, both $\matA_{S}$, $\matB_{S}$ have full column rank, and 
consequently the associated Gram matrices $\matA_{S}^{T}\matA_{S}$, $\matB_{S}^{T}\matB_{S}$ are positive definite. Further, 
since $\matA$ and $\matB$ have unit norm columns, both $\matA_{S}^{T}\matA_{S}$ and $\matB_{S}^{T}\matB_{S}$ are 
correlation matrices with unit diagonal entries. Using Proposition \ref{prop_gramian_kr_as_hadamard}, we can write
\begin{equation} \label{kr_as_hadamard_submat_form}
 (\matA_{S} \odot \matB_{S})^{T}(\matA_{S} \odot \matB_{S}) = \matA_{S}^{T}\matA_{S} \circ \matB_{S}^{T}\matB_{S}. 
\end{equation}

Next, for $k \le m$, by applying Lemma \ref{lem_matrix_kantorovich_ineq_2_hadamard_ver} to the positive definite matrices 
$(\matA_{S}^{T}\matA_{S})^{1/2}$ and $(\matB_{S}^{T}\matB_{S})^{1/2}$, we get,
\begin{eqnarray}
 \matA_{S}^{T}\matA_{S} \circ \matB_{S}^{T}\matB_{S} \!\!\!\!&\le& \!\!\!\!
 \lb \! (\matA_{S}^{T}\matA_{S})^{\frac{1}{2}} \! \circ\! (\matB_{S}^{T}\matB_{S})^{\frac{1}{2}} \! \rb^{2} 
 \!\! + \! \frac{1}{4}(M\!-\!m)^{2}\matI_{k}
 \nonumber \\
 &\le&   \matI_{k} + \frac{1}{4}(M-m)^{2}\matI_{k} \label{eq_MKI2H},
\end{eqnarray}
where the second inequality is a consequence of the unity bound on the spectral radius of the Hadamard product between correlation matrices, shown 
in Proposition \ref{prop_hadamard_corr_matrix_less_than_I}.
In \eqref{eq_MKI2H}, $M$ and $m$ are upper and lower bounds for the maximum and minimum eigenvalues of 
$(\matA_{S}^{T}\matA_{S})^{1/2} \otimes (\matB_{S}^{T}\matB_{S})^{1/2}$, respectively. 
From the restricted isometry of $\matA$ and $\matB$, and by application of Proposition \ref{kron_svd}, 
the minimum and maximum eigenvalues of $(\matA_{S}^{T}\matA_{S})^{1/2} \otimes (\matB_{S}^{T}\matB_{S})^{1/2}$ 
are lower and upper bounded by $\sqrt{(1 - \delta_{k}^{\matA})(1 - \delta_{k}^{\matB})}$ and $\sqrt{(1 + \delta_{k}^{\matA})(1 + \delta_{k}^{\matB})}$, respectively.
By introducing $\delta \triangleq \text{max}(\delta_{k}^{\matA}, \delta_{k}^{\matB})$, it is easy to check that the eigenvalues of 
$(\matA_{S}^{T}\matA_{S})^{1/2} \otimes (\matB_{S}^{T}\matB_{S})^{1/2}$ also lie inside the interval $\ls 1 - \delta, 1 + \delta \rs$.
Plugging $m = 1 - \delta$ and $M = 1+ \delta$ in \eqref{eq_MKI2H}, and by using \eqref{kr_as_hadamard_submat_form}, we get
\begin{equation} \label{thm_kr_rip_seed1}
 (\matA_{S} \odot \matB_{S})^{T}(\matA_{S} \odot \matB_{S})
 \le 
 \lb 1 + \delta^{2} \rb \matI_{k}.
\end{equation}
Similarly, by applying Lemma \ref{lem_reverse_kantorovich_hadamard_ver} to $\matA_{S}^{T}\matA_{S}$ and 
$\matB_{S}^{T}\matB_{S}$ with $m = 1 - \delta$ and $M = 1 + \delta$, we obtain
\begin{equation}  
 (\matA_{S}^{T}\matA_{S})^{1/2} \circ (\matB_{S}^{T}\matB_{S})^{1/2} \ge \lb \sqrt{1 - \delta^{2}} \rb \matI_{k}.
 \nonumber
\end{equation}
From Proposition~\ref{prop_hadamard_prod_pow_ineq}, we have $\matA_{S}^{T}\matA_{S} \circ \matB_{S}^{T}\matB_{S} \ge \lb (\matA_{S}^{T}\matA_{S})^{1/2} \circ (\matB_{S}^{T}\matB_{S})^{1/2} \rb^{2}$.
Therefore, we can write 
\begin{equation}  
 \matA_{S}^{T}\matA_{S} \circ \matB_{S}^{T}\matB_{S} \ge \lb 1 - \delta^{2} \rb \matI_{k}.
 \nonumber
\end{equation}
Further, using \eqref{kr_as_hadamard_submat_form}, we get
\begin{equation} \label{thm_kr_rip_seed2}
 (\matA_{S} \odot \matB_{S})^{T}(\matA_{S} \odot \matB_{S})
 \ge 
 \lb 1 - \delta^{2} \rb \matI_{k}.
\end{equation}
Finally, Theorem \ref{thm_ric_bound_for_kr}'s statement follows from \eqref{thm_kr_rip_seed1} and \eqref{thm_kr_rip_seed2}.

\section{Preliminaries for Probabilistic $k$-RIC Bound} \label{sec:bkgd_probabilistic_bnd}
In this section, we briefly discuss some concentration results for functions of 
subgaussian random variables which will appear in the proofs of 
Theorems~\ref{thm_kric_bound_X_neq_Y} and~\ref{thm_kric_bound_X_eq_Y}. 


\subsection{A Tail Probability for Subgaussian Vectors}
The theorem below presents the \textit{Hanson-Wright} inequality \cite{adamczak2015,Rudelson13HansonWrightIneq}, 
a tail probability for a quadratic form 
constructed using independent subgaussian random variables. 
\begin{theorem}[\textcolor{black}{Rudelson and Vershynin\cite{Rudelson13HansonWrightIneq}}] \label{thm_hanson_wright_subgaussian_conc}
Let $\vecx = (\vecx_{1}, \vecx_{2}, \ldots, \vecx_{n}) \in \Real^{n}$ be a random vector with independent components 
$\vecx_{i}$ satisfying $\EE \vecx_{i} = 0$ and $\lvv \vecx_{i} \rvv_{\psi_{2}} \le K$. Let $\matA$ be an $n \times n$ matrix. 
Then, for every $t \ge 0$,
\begin{eqnarray} \label{eqn_hanson_wright_ineq}
 && \hspace{-0.5cm} \PP \lc \lv \vecx^{T} \matA \vecx - \EE \vecx^{T} \matA \vecx \rv > t  \rc 
 \nonumber \\
 && \hspace{1cm} \le
 2 \exp{ \ls -c \min \lb \frac{t^2}{K^4 \lvv \matA \rvv_{HS}^{2}}, \frac{t}{K^2 \lvvv \matA \rvvv_{2}}  \rb \rs}
 \nonumber
\end{eqnarray}
where $c$ is a universal positive constant. 
\end{theorem}

The following corollary of Hanson-Wright inequality discusses the concentation of 
weighted inner product between two subgaussian vectors.
\begin{corollary} \label{corr_weighted_inner_product_tail}
Let $\vecu = (\vecu_{1}, \vecu_{2}, \ldots, \vecu_{n}) \in \Real^{n}$ and $\vecv = (\vecv_{1}, \vecv_{2}, \ldots, \vecv_{n}) \in \Real^{n}$ be independent random vectors with independent subgaussian components satisfying $\EE \vecu_{i} = \EE \vecv_{i} = 0$ and $\lvv \vecu_{i} \rvv_{\psi_{2}} \le K$, 
$\lvv \vecv_{i} \rvv_{\psi_{2}} \le K$. Let $\matD$ be an $n \times n$ matrix. 
Then, for every $t \ge 0$,
\begin{eqnarray} 
 && \hspace{-0.5cm} \PP \lc \lv \vecu^{T} \matD \vecv \rv > t  \rc 
 \nonumber \\
 && \hspace{1cm} \le
 2 \exp{ \ls -c \min \lb \frac{t^2}{K^4 \lvv \matD \rvv_{HS}^{2}}, \frac{t}{K^2 \lvvv \matD \rvvv_{2}}  \rb \rs}
 \nonumber
\end{eqnarray}
where $c$ is a universal positive constant. 
\end{corollary}
\begin{proof}
The desired tail bound is obtained by using the Hanson-Wright inequality in Theorem~\ref{thm_hanson_wright_subgaussian_conc} with 
$\vecx = \ls \vecu^{T} \vecv^{T} \rs^{T}$ and $\matA = \ls \mathbf{0}_{n\times n}\; \vert \;\matD; 
\mathbf{0}_{n\times n}\; \vert \;\mathbf{0}_{n\times n} \rs$. 
\end{proof}

\section{Proof of the Probabilistic $k$-RIC Bound (Theorem~\ref{thm_kric_bound_X_neq_Y})} \label{sec:proof_probabilistic}
The proof of Theorem~\ref{thm_kric_bound_X_neq_Y} starts with a variational definition of the $k$-RIC, 
$\delta_{k}\lb\frac{\matA}{\sqrt{m}} \odot \frac{\matB}{\sqrt{m}} \rb$ given below. 
\begin{equation} \label{variational_k_ric_char}
\delta_{k}\lb \frac{\matA}{\sqrt{m}} \odot \frac{\matB}{\sqrt{m}} \rb = 
\!\!\!\!\!\!\!\! \sup_{\substack{\vecz \in \Real^{n}, \\ \lvv \vecz \rvv_{2} = 1, \lvv \vecz \rvv_{0} \le k}} 
\!\!\lv \lvv \lb \frac{\matA}{\sqrt{m}} \odot \frac{\matB}{\sqrt{m}} \rb \vecz \rvv_{2}^{2} -1   \rv.
\end{equation}
In order to find a probabilistic upper bound for $\delta_{k}$, we intend to find a constant $\delta \in (0,1)$ such that 
$\PP(\delta_{k}\lb \frac{\matA}{\sqrt{m}} \odot \frac{\matB}{\sqrt{m}} \rb \ge \delta)$ 
is arbitrarily close to zero. We therefore consider the tail event 
\begin{equation}
\mathcal{E} \triangleq
 \lc  
 \sup_{\substack{\vecz \in \Real^{n}, \\ \lvv \vecz \rvv_{2} = 1, \lvv \vecz \rvv_{0} \le k}} 
 \lv  \lvv \lb \frac{\matA}{\sqrt{m}} \odot \frac{\matB}{\sqrt{m}} \rb \vecz \rvv_{2}^{2} - 1 \rv \ge \delta
 \rc,
 \label{var_k_char_2}
\end{equation}
and show that for $m$ sufficiently large, $\PP(\mathcal{E})$ can be driven arbitrarily close to zero. 
In other words, the constant $\delta$ serves as a probabilistic upper bound for 
$\delta_{k} \lb \frac{\matA}{\sqrt{m}} \odot \frac{\matB}{\sqrt{m}} \rb$.
Let $\setU_{k}$ denote the set of all $k$ or less sparse unit norm vectors in $\Real^{n}$. Then, using Proposition~\ref{prop_gramian_kr_as_hadamard}, the tail event in \eqref{var_k_char_2} can be rewritten as 
\begin{eqnarray}
\PP(\mathcal{E}) &=&  \PP \lb  
\underset{\vecz \in \setU_{k}}{\sup} \lv \vecz^{T} \lb \matA \odot \matB \rb^{T} \lb \matA \odot \matB \rb \vecz  - m^2 \rv \ge \delta m^2
\rb 
\nonumber \\
&& \hspace{-1.5cm} =  \PP \lb  
 \underset{\vecz \in \setU_{k}}{\sup} \lv \vecz^{T} \lb \matA^{T} \matA \circ \matB^{T} \matB \rb  \vecz  - m^2 \rv \ge \delta m^2
 \rb 
 \nonumber \\
&& \hspace{-1.5cm} =
\PP \lb  \underset{\vecz \in \setU_{k}}{\sup} \lv \sum_{i=1}^{n} \sum_{j=1}^{n} z_{i} z_{j} \lb \veca_{i}^{T}\veca_{j} \rb \lb \vecb_{i}^{T}\vecb_{j} \rb  - m^2 \rv \ge \delta m^2
  \rb, 
\end{eqnarray}
where $\veca_{i}$ and $\vecb_{i}$ denote the $i$th column of $\matA$ and $\matB$, respectively. Further, by applying the triangle inequality and the union bound, the above tail probability splits as
\begin{eqnarray}
\PP(\mathcal{E})
&\le&  
 \PP \lb  
\underset{\vecz \in \setU_{k}}{\sup} \lv \sum_{i=1}^{n} z_{i}^{2} \lvv \veca_{i} \rvv^{2}_{2} \lvv \vecb_{i} \rvv_{2}^{2} - m^2 \rv \ge \alpha \delta m^2
\rb 
\nonumber \\
&& \hspace{-1.8cm} + \PP \lb  
\underset{\vecz \in \setU_{k}}{\sup} \lv \sum_{i=1}^{n} \sum_{\substack{j=1, j \neq i}}^{n} 
 z_{i} z_{j}  \veca_{i}^{T}\veca_{j}  \vecb_{i}^{T}\vecb_{j} \rv \ge (1-\alpha) \delta m^2 \!
 \rb. 
\label{eqn_first_split}
\end{eqnarray}
In the above, $\alpha \in (0,1)$ is a variational union bound parameter which can be optimized at a later stage. 
We now proceed to find separate upper bounds for each of the two probability 
terms in \eqref{eqn_first_split}.

The first probability term in \eqref{eqn_first_split} admits the following sequence of relaxations. 
\begin{eqnarray}
&&  \hspace{-0.5cm} \PP \lb  
\underset{\vecz \in \setU_{k}}{\sup} \lv \sum_{i=1}^{n} z_{i}^{2} \lvv \veca_{i} \rvv^{2}_{2} \lvv \vecb_{i} \rvv_{2}^{2} - m^2 \rv \ge \alpha \delta m^2
\rb 
\nonumber \\
&\stackrel{(a)}{\le}&  \PP \lb  
\underset{\vecz \in \setU_{k}}{\sup} \sum_{i=1}^{n} z_{i}^{2} \lv \lvv \veca_{i} \rvv^{2}_{2} \lvv \vecb_{i} \rvv_{2}^{2} - m^2 \rv \ge \alpha \delta m^2
\rb 
\nonumber \\
&\stackrel{(b)}{\le}&  \PP \lb  
\underset{1 \le i \le n}{\max} \lv \lvv \veca_{i} \rvv^{2}_{2} \lvv \vecb_{i} \rvv_{2}^{2} - m^2 \rv \ge \alpha \delta m^2
\rb 
\nonumber \\
&\stackrel{(c)}{=}& \PP \lb  
\bigcup_{1 \le i \le n} \lc \lv \lvv \veca_{i} \rvv^{2}_{2} \lvv \vecb_{i} \rvv_{2}^{2} - m^2 \rv \ge \alpha \delta m^2 \rc
\rb 
\nonumber \\
& \stackrel{(d)}{=} & 
n \PP \lb \lv \lvv \veca_{1} \rvv^{2}_{2} \lvv \vecb_{1} \rvv_{2}^{2} - m^2 \rv \ge \alpha \delta m^2 \rb
\nonumber \\
&\stackrel{(e)}{\le}& 
n \PP \lb \lv \lvv \veca_{1} \rvv^{2}_{2} - m \rv \lv \lvv \vecb_{1} \rvv^{2}_{2} - m \rv \ge \alpha \beta \delta m^2 \rb 
\nonumber \\
&& 
+ 2n \PP \lb \lv  \lvv \veca_{1} \rvv^{2}_{2} - m \rv \ge \frac{\alpha (1 - \beta) \delta m}{2} \rb
\nonumber \\
&\stackrel{(f)}{\le}& 
2n \PP \lb \lv \lvv \veca_{1} \rvv^{2}_{2} - m \rv \ge \sqrt{\alpha \beta \delta} m \rb 
\nonumber \\
&& 
+ 2n \PP \lb \lv  \lvv \veca_{1} \rvv^{2}_{2} - m \rv \ge \frac{\alpha (1 - \beta) \delta m}{2} \rb.
\nonumber \\
&\stackrel{(g)}{\le}& 
4n \PP \lb \lv \lvv \veca_{1} \rvv^{2}_{2} - m \rv \ge \frac{\alpha \delta m}{2} \lb 1 - \frac{\alpha \delta}{4} \rb \rb 
\nonumber \\
&\stackrel{(h)}{\le}& 
 8n  e^{-c m \frac{\alpha^{2} \delta^{2}}{4 \kappa_{o}^{4}} (1 - \alpha \delta/4)^{2} } 
 \nonumber \\
 &=& 8 n^{-\lb   \frac{c m \alpha^{2} \delta^{2} (1 - \alpha \delta/4)^{2} }{4 \kappa_{o}^{4} \log n}  -1 \rb}.
\label{eqn_intr1}
\end{eqnarray}
In the above, step ($a$) follows from the triangle inequality combined with the fact that $z_{i}^{2}$'s sum to one. 
The inequality in step ($b$) is a consequence of the fact that any nonnegative convex combination of $n$ arbitrary numbers 
is at most the maximum among the $n$ numbers. Step ($c$) is obtained by simply rewriting the tail event for the maximum of $n$ random variables 
as the union of tail events for the individual random variables. Step ($d$) is the application of 
the union bound over values of index $i \in [n]$ and exploiting the i.i.d. nature of the columns of $\matA$ and $\matB$. Step ($e$) is the union bound combined with the fact that for any two vectors $\veca, \vecb \in \Real^{m}$, the following triangle inequality holds: 
\begin{eqnarray}
 \lv \lvv \veca \rvv^{2}_{2} \lvv \vecb \rvv_{2}^{2} - m^2 \rv 
&\le& \lv \lb \lvv \veca \rvv^{2}_{2} - m \rb \lb \lvv \vecb \rvv^{2}_{2} - m \rb \rv 
\nonumber \\
&& + m \lv  \lvv \veca \rvv^{2}_{2} - m \rv  
 + m \lv  \lvv \vecb \rvv^{2}_{2} - m \rv.
 \nonumber
\end{eqnarray}
In step ($e$), $\beta \in (0,1)$ is a variational union bound parameter. 
Step ($f$) is once again the union bound which exploits the fact that the columns 
$\veca_{1}$ and $\vecb_{1}$ are identically distributed. Step ($g$) is obtained 
by setting $\beta = \alpha \delta /4$. 
Lastly, step ($h$) is the Hanson-Wright inequality (Theorem~\ref{thm_hanson_wright_subgaussian_conc}) 
applied to the subgaussian vector~$\veca_{1}$.

Next, we turn our attention to the second tail probability term in \eqref{eqn_first_split}. We note that
\begin{eqnarray}
&& \hspace{-0.7cm}
 \underset{\vecz \in \setU_{k}}{\sup} \lv \sum_{i=1}^{n} \sum_{j=1, j \neq i}^{n} z_{i} z_{j} \veca_{i}^{T}\veca_{j} \vecb_{i}^{T}\vecb_{j}  \rv 
\nonumber \\
&& \hspace{-0.1cm} \le 
\underset{\vecz \in \setU_{k}}{\sup} \sum_{i=1}^{n} 
 \sum_{j =1, j \neq i}^{n} \lv z_{i} z_{j} \rv \lv \veca_{i}^{T}\veca_{j} \rv \lv \vecb_{i}^{T}\vecb_{j} 
 \rv 
\nonumber \\
&& \hspace{-0.1cm} \le 
\underset{\vecz \in \setU_{k}}{\sup} \! \lb 
\sum_{i=1}^{n} \sum_{j =1, j \neq i}^{n} \! \lv z_{i} z_{j} \rv \rb \!
\lb \max_{\substack{i, j \in \text{supp}(\vecu), \\  i \neq j}}  \!
\lv \veca_{i}^{T}\veca_{j} \rv \lv \vecb_{i}^{T}\vecb_{j} \rv \rb
\nonumber \\
&& \hspace{-0.1cm} \le 
k \lb \max_{\substack{i, j \in [n], \\  i \neq j}}  \!
\lv \veca_{i}^{T}\veca_{j} \rv \lv \vecb_{i}^{T}\vecb_{j} \rv \rb,
\label{eqn_intr2}
\end{eqnarray}
where the second step is the application of the H\"{o}lders inequality. The last step uses the 
fact that $\lvv \vecz \rvv_{1} \le \sqrt{k}$ for $\vecz \in \setU_{k}$. 
Using \eqref{eqn_intr2}, and by applying the union bound over $\binom{n}{2}$ possible distinct $(i, j)$ pairs, the second probability term in \eqref{eqn_first_split} can be bounded as
\begin{eqnarray}
 && \PP \lb  
\underset{\vecz \in \setU_{k}}{\sup} \lv \sum_{i=1}^{n} \sum_{j=1, j \neq i}^{n} z_{i} z_{j} \veca_{i}^{T}\veca_{j} \vecb_{i}^{T}\vecb_{j} \rv \ge (1-\alpha) \delta m^2
\rb
\nonumber \\
&& \hspace{1cm} \le 
\frac{n^2}{2} \PP \lb \lv \veca_{1}^{T}\veca_{2} \rv \lv \vecb_{1}^{T}\vecb_{2} \rv  \ge \frac{(1-\alpha)\delta m^2}{k} \rb 
\nonumber \\
&& \hspace{1cm} \le 
n^2 \PP \lb \lv \veca_{1}^{T}\veca_{2} \rv \ge \frac{\sqrt{(1-\alpha)\delta} m}{\sqrt{k}} \rb
\nonumber \\
&& \hspace{1cm} \le 
2n^2 e^{-\frac{c (1- \alpha) \delta m}{\kappa_{o}^4 k}} = 2 n^{-\lb \frac{c (1-\alpha) \delta m}{\kappa_{o}^4 k \log n}  - 2\rb}.    
\label{eqn_intr4}
\end{eqnarray}
The last inequality in the above is obtained by using the tail bound for $|\veca_{1}^{T}\veca_{2}|$ 
from Corollary~\ref{corr_weighted_inner_product_tail}.
Finally, by combining \eqref{eqn_first_split}, \eqref{eqn_intr1} and \eqref{eqn_intr4}, and setting 
$\alpha = 1/2$, we obtain the following simplified tail bound, 
\begin{equation}
\PP(\mathcal{E}) \le 
8 n^{- \lb \frac{c m \delta^2 (1 - \delta/8)^2}{16 \kappa_{o}^{4} \log n} -1 \rb}
+ 2 n^{- \lb \frac{c \delta m}{ 2 \kappa_{o}^4 k \log n} - 2\rb}.
\label{eqn_intr144}
\end{equation}
From \eqref{eqn_intr144}, for $m > 
\max{ \lb \frac{4 \gamma \kappa_{o}^2 k \log n}{c \delta} , \frac{32 \gamma \kappa_{o}^4 \log n}{c \delta^2 (1 - \delta/8)^2} \rb} $ and any $\gamma > 1$, we have $\PP(\mathcal{E}) < 10 n^{-2(\gamma -1)}$. Note that, in terms of $k$ and $n$, the
first term in the inequality for $m$ scales as
$k \log n$; it dominates the second term, which scales as $\log{n}$. This ends our proof.

\section{Conclusions} \label{sec:concluding_remarks}
In this work, we have analyzed the restricted isometry property of the columnwise 
Khatri-Rao product matrix in terms of its restricted isometry constants. 
We gave two upper bounds for the $k$-RIC of a generic columnwise Khatri-Rao product matrix.
The first $k$-RIC bound, a deterministic bound, is valid for the Khatri-Rao product of an 
arbitrary pair of input matrices of the same size with normalized columns. It is conveniently 
computed in terms of the $k$-RICs of the input matrices. 
We also gave a probabilistic RIC bound for the columnwise KR product of a 
pair of random matrices with i.i.d. subgaussian entries. The probabilistic RIC 
bound is one of the key components needed for computing sample complexity bounds 
for several machine learning algorithms. 

The analysis of the RIP of Khatri-Rao product matrices in this article can be extended in multiple ways. 
The current RIC bounds can be extended to the Khatri-Rao product of three or more matrices.
More importantly, in order to relate the RICs to the dimensions of the input matrices, we had to resort to the randomness in their entries. Removing this randomness aspect of our results 
could be an interesting direction for future work.


\appendix

\subsection{Proof of Proposition \ref{corr_mat_sqrt_hadamard_is_double_stochastic}}
\label{App:proof_corr_mat_sqrt_hadamard_is_double_stochastic}
\begin{proof}
Since $\matA$ is a correlation matrix, it admits the Schur decomposition, $\matA = \matU \mathbf{\Lambda} \matU^{T}$, 
with unitary $\matU$ and eigenvalue matrix $\mathbf{\Lambda} = \diag{\lb \lambda_{1}, \lambda_{2}, \ldots, \lambda_{n}\rb}$. 
Since $\matA$ is positive semi-definite, its nonnegative square-root exists and 
is given by $\matA^{1/2} = \matU \Lambda^{1/2} \matU^{T}$. Consider
\begin{eqnarray}
&& \hspace{-0.9cm}
\matA^{1/2} \circ \matA^{1/2} = 
 \sum_{i = 1}^{n} \lambda_{i}^{1/2}\vecu_{i} \vecu_{i}^{T}  
\circ
 \sum_{j = 1}^{n} \lambda_{j}^{1/2}\vecu_{j} \vecu_{j}^{T}  
\nonumber \\
&& \hspace{1cm} =
\sum_{i = 1}^{n} \sum_{j = 1}^{n} \lambda_{i}^{1/2} \lambda_{j}^{1/2} \lb \vecu_{i} \vecu_{i}^{T} \rb 
\circ \lb \vecu_{j} \vecu_{j}^{T} \rb
\nonumber \\
&& \hspace{1cm} =
\sum_{i = 1}^{n} \sum_{j = 1}^{n} \lambda_{i}^{1/2} \lambda_{j}^{1/2} \lb \vecu_{i} \circ \vecu_{j} \rb 
\lb \vecu_{i} \circ \vecu_{j} \rb^{T}.
\label{eqn_sqrt_hadamard_expr}
\end{eqnarray}

The second equality above follows from the distributive property of the Hadamard product and the last step follows from Fact~$7.6.2$ in
\cite{Bernstein09MatrixMath}.
Using \eqref{eqn_sqrt_hadamard_expr}, we can show that the rows and columns of $\matA^{1/2} \circ \matA^{1/2}$ sum to one,
as follows:
\begin{eqnarray}
&& \hspace{-1cm} 
(\matA^{1/2} \circ \matA^{1/2}) \mathbf{1} = 
 \sum_{i = 1}^{n} \sum_{j = 1}^{n} \lambda_{i}^{1/2} \lambda_{j}^{1/2} \lb \vecu_{i} \circ \vecu_{j} \rb 
\lb \vecu_{i} \circ \vecu_{j} \rb^{T} \mathbf{1}
\nonumber \\
&& \hspace{1.5cm} =
 \sum_{i = 1}^{n} \sum_{j = 1}^{n} \lambda_{i}^{1/2} \lambda_{j}^{1/2} \lb \vecu_{i} \circ \vecu_{j} \rb 
\lb \vecu_{i} \circ \mathbf{1} \rb^{T} \vecu_{j}
\nonumber \\
&& \hspace{1.5cm} =
 \sum_{i = 1}^{n} \sum_{j = 1}^{n} \lambda_{i}^{1/2} \lambda_{j}^{1/2} \lb \vecu_{i} \circ \vecu_{j} \rb 
\vecu_{i}^{T} \vecu_{j}
\nonumber \\
&& \hspace{1.5cm} =
 \sum_{i = 1}^{n} \lambda_{i} \lb \vecu_{i} \circ \vecu_{i} \rb = \vecd \;\; (\text{say}). 
\nonumber
\end{eqnarray}

The above arguments follow from the orthonormality of the columns of $\matU$, and 
repeated application of Fact~$7.6.1$ in \cite{Bernstein09MatrixMath}.
Note that for $k \in [n]$, $\vecd(k) = \sum_{i=1}^{n} \lambda(i) \lb\vecu_{i}(k)\rb^{2} = [\matU \mathbf{\Lambda} \matU^{T}]_{kk} 
= \matA_{kk} = 1$. Thus, we have shown that $(\matA^{1/2} \circ \matA^{1/2}) \mathbf{1} = \mathbf{1}$. 
Likewise, it can be shown that $ \mathbf{1}^{T} (\matA^{1/2} \circ \matA^{1/2}) = \mathbf{1}^{T}$.
Thus, $\matA^{1/2} \circ \matA^{1/2}$ is doubly stochastic.
\end{proof}



\subsection{Proof of Theorem \ref{thm_kric_bound_X_eq_Y}} \label{App:proof_thm_kric_bound_X_eq_Y}
\begin{proof}
The proof of Theorem~\ref{thm_kric_bound_X_eq_Y} follows along similar lines as that of Theorem~\ref{thm_kric_bound_X_neq_Y}. We consider the tail event
\begin{equation}
\mathcal{E}_{1} \triangleq
 \lc  
 \sup_{\vecz \in \setU_{k}} 
 \lv  \lvv \lb \frac{\matA}{\sqrt{m}} \odot \frac{\matA}{\sqrt{m}} \rb \vecz \rvv_{2}^{2} - 1 \rv \ge \delta
 \rc
 \label{var_k_rip_char},
\end{equation}
and show that for sufficiently large $m$, $\PP(\mathcal{E}_{1})$ can be driven arbitrarily close to zero, thereby 
implying that $\delta$ is a probabilistic upper bound for $\delta_{k}\lb (\matA/\sqrt{m} \odot \matA/\sqrt{m}) \rb$. Once again, $\setU_{k}$ denotes the set of all $k$ or less sparse unit norm vectors in $\Real^{m}$. We note that $\PP(\mathcal{E}_{1})$ admits the following union bound:
\begin{eqnarray}
\PP(\mathcal{E}_{1}) \!\!\! &=& \!\!\! \PP \lb  
\underset{\vecz \in \setU_{k}}{\sup} \lv \vecz^{T} \lb \matA \odot \matA \rb^{T} \lb \matA \odot \matA \rb \vecz  - m^2 \rv \ge \delta m^2
\rb 
\nonumber \\
&& \hspace{-1cm} = 
\PP \lb  
 \underset{\vecz \in \setU_{k}}{\sup} \lv \vecz^{T} \lb \matA^{T} \matA \circ \matA^{T} \matA \rb  \vecz  - m^2 \rv \ge \delta m^2
 \rb 
 \nonumber \\
&& \hspace{-1cm} =
 \PP \lb  
 \underset{\vecz \in \setU_{k}}{\sup} \lv \sum_{i=1}^{n} \sum_{j=1}^{n} z_{i} z_{j} \lb \veca_{i}^{T}\veca_{j} \rb^{2}  - m^2 \rv \ge \delta m^2
 \rb 
 \nonumber \\ 
&& \hspace{-1cm} \le 
 \PP \lb  
\underset{\vecz \in \setU_{k}}{\sup} \lv \sum_{i=1}^{n} z_{i}^{2} \lvv \veca_{i} \rvv^{4}_{2} - m^2 \rv \ge \alpha \delta m^2
\rb 
\nonumber \\
&&  \hspace{-1cm}
+ \PP \lb  
\underset{\vecz \in \setU_{k}}{\sup} \lv \sum_{i=1}^{n} \sum_{\substack{j=1, j \neq i}}^{n} \!\!\!\! z_{i} z_{j} \lb \veca_{i}^{T}\veca_{j} \rb^{2} \rv \ge (1-\alpha) \delta m^2
\rb. 
\label{eqn_first_split_2}
\end{eqnarray}
In the above, the second identity follows from Proposition~\ref{prop_gramian_kr_as_hadamard}.
The last inequality uses the triangle inequality followed by the union bound, with $\alpha \in (0,1)$ being a variational parameter 
to be optimized later. Similar to the proof of Theorem~\ref{thm_kric_bound_X_neq_Y}, we now derive separate upper bounds for each of the two probability 
terms in \eqref{eqn_first_split_2}.

The first term in \eqref{eqn_first_split_2} admits the following series of relaxations.
\begin{eqnarray}
&&  \hspace{-1.5cm} \PP \lb  
\underset{\vecz \in \setU_{k}}{\sup} \lv \sum_{i=1}^{n} z_{i}^{2} \lvv \veca_{i} \rvv^{4}_{2}  - m^2 \rv \ge \alpha \delta m^2
\rb 
\nonumber \\
&\stackrel{(a)}{\le}&  \PP \lb  
\underset{\vecz \in \setU_{k}}{\sup} \sum_{i=1}^{n} z_{i}^{2} \lv \lvv \veca_{i} \rvv^{4}_{2} - m^2 \rv \ge \alpha \delta m^2
\rb 
\nonumber \\
&\stackrel{(b)}{\le}&  \PP \lb  
\underset{1 \le i \le n}{\max} \lv \lvv \veca_{i} \rvv^{4}_{2} - m^2 \rv \ge \alpha \delta m^2
\rb 
\nonumber \\
& \stackrel{(c)}{\le} & 
n \PP \lb \lv \lvv \veca_{1} \rvv^{4}_{2} - m^2 \rv \ge \alpha \delta m^2 \rb. 
\nonumber \\
& \stackrel{(d)}{\le}&
n \PP \lb \lv \lvv \veca_{1} \rvv^{2}_{2} - m \rv^{2}  \ge \alpha \beta \delta m^2 \rb 
\nonumber \\
&&
+ n \PP \lb \lv  \lvv \veca_{1} \rvv^{2}_{2} - m \rv \ge \frac{\alpha (1 - \beta) \delta m}{2} \rb
\nonumber \\
&\stackrel{(e)}{\le}& 
n \PP \lb \lv \lvv \veca_{1} \rvv^{2}_{2} - m \rv \ge \sqrt{\alpha \beta \delta} m \rb 
\nonumber \\
&&
+ n \PP \lb \lv  \lvv \veca_{1} \rvv^{2}_{2} - m \rv \ge \frac{\alpha (1 - \beta) \delta m}{2} \rb.
\nonumber \\
&\stackrel{(f)}{\le}& 
2n \PP \lb \lv \lvv \veca_{1} \rvv^{2}_{2} - m \rv \ge \frac{\alpha \delta m}{2} \lb 1 - \frac{\alpha \delta}{4} \rb \rb 
\nonumber \\
&\stackrel{(g)}{\le}& 
 4n  e^{-c m \frac{\alpha^{2} \delta^{2}}{4 \kappa_{o}^4} (1 - \alpha \delta/4)^{2} }
 \nonumber \\
 &=& 4 n^{-\lb \frac{c m \alpha^{2} \delta^{2} (1 - \alpha \delta/4)^2 }{4 \kappa_{o}^4  \log n} -1\rb}.
\label{eqn_intr12}
\end{eqnarray}
In the above, step ($a$) is the triangle inequality. The inequality 
in step ($b$) follows from the fact that nonnegative convex combination of 
$n$ arbitrary numbers is at most the maximum among the $n$ numbers. 
Step ($c$) is a union bound. Step ($d$) is also a union bounding argument with $\beta \in (0,1)$ as 
a variational parameter, combined 
with the fact that for any vector $\veca$, the triangle inequality 
$ \lv \lvv \veca \rvv^{4}_{2} - m^2 \rv \le  \lv \lvv \veca \rvv^{2}_{2} - m \rv^{2} + 2m \lv  \lvv \veca \rvv^{2}_{2} - m \rv$
is always true. 
Step ($f$) is obtained by choosing the union bound parameter $\beta = \alpha \delta /4$.
Finally, step ($g$) is the Hanson-Wright inequality 
(Theorem~\ref{thm_hanson_wright_subgaussian_conc}) applied to the subgaussian vector $\veca_{1}$.

Next, we derive an upper bound for the second tail probability term in \eqref{eqn_first_split}. We observe that
\begin{eqnarray}
\underset{\vecz \in \setU_{k}}{\sup} \lv \sum_{i=1}^{n} \sum_{j=1, j \neq i}^{n} z_{i} z_{j} \lb \veca_{i}^{T}\veca_{j} \rb^2  \rv && 
\nonumber \\
&& \hspace{-4cm} \le 
\underset{\vecz \in \setU_{k}}{\sup} \sum_{i=1}^{n} 
 \sum_{j =1, j \neq i}^{n} \lv z_{i} z_{j} \rv \lb \veca_{i}^{T}\veca_{j} \rb^2 
\nonumber \\
&& \hspace{-4cm} \le 
\underset{\vecz \in \setU_{k}}{\sup} \lb \sum_{i=1}^{n} \sum_{j =1, j \neq i}^{n} \lv z_{i} z_{j} \rv \rb
\lb \max_{\substack{i, j \in \text{supp}(\vecu), \\  i \neq j}}  
\lv \veca_{i}^{T}\veca_{j} \rv^2 \rb
\nonumber \\
&& \hspace{-4cm} \le 
k \lb \max_{\substack{i, j \in [n], \\  i \neq j}}  
\lv \veca_{i}^{T}\veca_{j} \rv^2 \rb,
\label{eqn_intr2_2}
\end{eqnarray}
where the second step is the application of the H\"{o}lders inequality. The last step uses the 
fact that $\lvv \vecz \rvv_{1} \le \sqrt{k}$ for $\vecz \in \setU_{k}$. 
Using \eqref{eqn_intr2_2}, and by applying the union bound over $\binom{n}{2}$ possible distinct $(i, j)$ pairs, the second probability term in \eqref{eqn_first_split} can be bounded as
\begin{eqnarray}
 && \PP \lb  
\underset{\vecz \in \setU_{k}}{\sup} \lv \sum_{i=1}^{n} \sum_{j=1, j \neq i}^{n} z_{i} z_{j} \lb \veca_{i}^{T}\veca_{j} \rb^2 \rv \ge (1-\alpha) \delta m^2 
\rb 
\nonumber \\
&& \hspace{1cm} \le 
\frac{n^2}{2} \PP \lb \lv \veca_{1}^{T}\veca_{2} \rv^2  \ge \frac{(1-\alpha)\delta m^2}{k} \rb 
\nonumber \\
&& \hspace{1cm} \le 
n^2 \PP \lb \lv \veca_{1}^{T}\veca_{2} \rv \ge \frac{\sqrt{(1-\alpha)\delta} m}{\sqrt{k}} \rb
\nonumber \\
&& \hspace{1cm} \le 
n^2 e^{-\frac{c (1- \alpha) \delta m}{\kappa_{o}^4 k}} 
= n^{-\lb \frac{c (1-\alpha) \delta m}{\kappa_{o}^4 k \log n}  - 2\rb}.    
\label{eqn_intr60}
\end{eqnarray}
The last inequality in the above is obtained by using the tail bound for $|\veca_{1}^{T}\veca_{2}|$ 
from Corollary~\ref{corr_weighted_inner_product_tail}.
Finally, by combining \eqref{eqn_first_split_2}, \eqref{eqn_intr12} and \eqref{eqn_intr60}, and setting 
$\alpha = 1/2$, we obtain the following tail bound. 
\begin{equation}
\PP(\mathcal{E}) \le 
4 n^{- \lb \frac{c m \delta^2 (1 - \delta/8)^2}{16 \kappa_{o}^{4} \log n} -1 \rb}
+ n^{- \lb \frac{c \delta m}{ 2 \kappa_{o}^4 k \log n} - 2\rb}.
\label{eqn_intr44}
\end{equation}
From \eqref{eqn_intr44}, we can conclude that for $m > 
\max{ \lb \frac{4 \gamma \kappa_{o}^4 k \log n}{c \delta} , \frac{32 \gamma \kappa_{o}^4 \log n}{c \delta^2 (1 - \delta/8)^2} \rb} $ and any $\gamma > 1$, we have $\PP(\mathcal{E}) < 5 n^{-2(\gamma -1)}$. Note that, in terms of $k$ and $n$, the
first term in the inequality for $m$ scales as
$\frac{k \log n}{\delta}$; it dominates the second term, which scales as $\log{n}$.
\end{proof}


\bibliographystyle{IEEEtran}
\bibliography{IEEEabrv,bibJournalList,Prob_KR_RIC_bounds}

\begin{thebibliography}{10}
\providecommand{\url}[1]{#1}
\csname url@samestyle\endcsname
\providecommand{\newblock}{\relax}
\providecommand{\bibinfo}[2]{#2}
\providecommand{\BIBentrySTDinterwordspacing}{\spaceskip=0pt\relax}
\providecommand{\BIBentryALTinterwordstretchfactor}{4}
\providecommand{\BIBentryALTinterwordspacing}{\spaceskip=\fontdimen2\font plus
\BIBentryALTinterwordstretchfactor\fontdimen3\font minus
  \fontdimen4\font\relax}
\providecommand{\BIBforeignlanguage}[2]{{%
\expandafter\ifx\csname l@#1\endcsname\relax
\typeout{** WARNING: IEEEtran.bst: No hyphenation pattern has been}%
\typeout{** loaded for the language `#1'. Using the pattern for}%
\typeout{** the default language instead.}%
\else
\language=\csname l@#1\endcsname
\fi
#2}}
\providecommand{\BIBdecl}{\relax}
\BIBdecl

\bibitem{KhatriRao68}
C.~G. Khatri and C.~R. Rao, ``Solutions to some functional equations and their
  applications to characterization of probability distributions,''
  \emph{Sankhya: The Indian Journal of Statistics, Series A (1961-2002)},
  vol.~30, no.~2, pp. 167--180, 1968.

\bibitem{Bernstein09MatrixMath}
D.~S. Bernstein, \emph{Matrix Mathematics: Theory, Facts, and Formulas (Second
  Edition)}.\hskip 1em plus 0.5em minus 0.4em\relax Princeton University Press,
  2009.

\bibitem{LiuTrenkler}
S.~Liu and G.~Trenkler, ``Hadamard, {K}hatri-{R}ao, {K}ronecker and other
  matrix products,'' \emph{International Journal of Information and System
  Sciences}, vol.~4, no.~1, pp. 160--177, 2007.

\bibitem{Duarte12KroneckerCS}
M.~F. Duarte and R.~G. Baraniuk, ``Kronecker compressive sensing,''
  \emph{{IEEE} Trans. Image Process.}, vol.~21, no.~2, pp. 494--504, Feb. 2012.

\bibitem{PPal15SrcLocalization}
A.~Koochakzadeh and P.~Pal, ``Sparse source localization in presence of
  co-array perturbations,'' in \emph{International Conference on Sampling
  Theory and Applications (SampTA)}, May 2015, pp. 563--567.

\bibitem{Romero16CovSense}
D.~Romero, D.~D. Ariananda, Z.~Tian, and G.~Leus, ``Compressive covariance
  sensing: Structure-based compressive sensing beyond sparsity,'' \emph{{IEEE}
  Signal Process. Mag.}, vol.~33, no.~1, pp. 78--93, Jan. 2016.

\bibitem{Dasarathy15CovMatEst}
G.~Dasarathy, P.~Shah, B.~N. Bhaskar, and R.~D. Nowak, ``Sketching sparse
  matrices, covariances, and graphs via tensor products,'' \emph{{IEEE} Trans.
  Inf. Theory}, vol.~61, no.~3, pp. 1373--1388, Mar. 2015.

\bibitem{Ma10DoaEstViaKR}
W.~K. Ma, T.~H. Hsieh, and C.~Y. Chi, ``{DOA} estimation of quasi-stationary
  signals with less sensors than sources and unknown spatial noise covariance:
  A {K}hatri-{R}ao subspace approach,'' \emph{{IEEE} Trans. Signal Process.},
  vol.~58, no.~4, pp. 2168--2180, Apr. 2010.

\bibitem{Sidiropoulos12CSTensors}
N.~D. Sidiropoulos and A.~Kyrillidis, ``Multi-way compressed sensing for sparse
  low-rank tensors,'' \emph{{IEEE} Signal Process. Lett.}, vol.~19, no.~11, pp.
  757--760, Nov. 2012.

\bibitem{Donoho06StableRecovery}
D.~L. Donoho, M.~Elad, and V.~N. Temlyakov, ``Stable recovery of sparse
  overcomplete representations in the presence of noise,'' \emph{{IEEE} Trans.
  Inf. Theory}, vol.~52, no.~1, pp. 6--18, Jan. 2006.

\bibitem{FoucartRauhut13CSBook}
S.~Foucart and H.~Rauhut, \emph{A Mathematical Introduction to Compressive
  Sensing}.\hskip 1em plus 0.5em minus 0.4em\relax Birkh{\"a}user Basel, 2013.

\bibitem{CandesTaoRomberg06StableRecovery}
E.~J. Cand\`{e}s, J.~K. Romberg, and T.~Tao, ``Stable signal recovery from
  incomplete and inaccurate measurements,'' \emph{Communications on Pure and
  Applied Mathematics}, vol.~59, no.~8, pp. 1207--1223, 2006.

\bibitem{CandesTao05}
E.~J. Candes and T.~Tao, ``Decoding by linear programming,'' \emph{{IEEE}
  Trans. Inf. Theory}, vol.~51, no.~12, pp. 4203--4215, Dec 2005.

\bibitem{Baraniuk08ASimpleProof}
R.~Baraniuk, M.~Davenport, R.~DeVore, and M.~Wakin,
  ``\BIBforeignlanguage{English}{A simple proof of the restricted isometry
  property for random matrices},''
  \emph{\BIBforeignlanguage{English}{Constructive Approximation}}, vol.~28,
  no.~3, pp. 253--263, 2008.

\bibitem{Cai10RICbound}
T.~T. Cai, L.~Wang, and G.~Xu, ``New bounds for restricted isometry
  constants,'' \emph{{IEEE} Trans. Inf. Theory}, vol.~56, no.~9, pp.
  4388--4394, Sept 2010.

\bibitem{PPal15CovMMV}
P.~Pal and P.~P. Vaidyanathan, ``Pushing the limits of sparse support recovery
  using correlation information,'' \emph{{IEEE} Trans. Signal Process.},
  vol.~63, no.~3, pp. 711--726, Feb. 2015.

\bibitem{WipfRao07MSBL}
D.~P. Wipf and B.~D. Rao, ``An empirical {B}ayesian strategy for solving the
  simultaneous sparse approximation problem,'' \emph{IEEE Transactions on
  Signal Processing}, vol.~55, no.~7, pp. 3704--3716, Jul. 2007.

\bibitem{KhannaCRM17MSBLSuffCond}
\BIBentryALTinterwordspacing
S.~Khanna and C.~R. Murthy, ``On the support recovery of jointly sparse
  {G}aussian sources using sparse {B}ayesian learning,'' \emph{CoRR}, vol.
  abs/1703.04930. [Online]. Available: \url{http://arxiv.org/abs/1703.04930}
\BIBentrySTDinterwordspacing

\bibitem{Chepuri17GraphSamp}
S.~P. Chepuri and G.~Leus, ``Graph sampling for covariance estimation,''
  \emph{IEEE Transactions on Signal and Information Processing over Networks},
  vol.~3, no.~3, pp. 451--466, Sept 2017.

\bibitem{Sidiropoulous12PARAFAC}
N.~D. Sidiropoulos and A.~Kyrillidis, ``Multi-way compressed sensing for sparse
  low-rank tensors,'' \emph{IEEE Signal Processing Letters}, vol.~19, no.~11,
  pp. 757--760, Nov. 2012.

\bibitem{Sidiropolos09MIMORadar}
D.~Nion and N.~D. Sidiropoulos, ``A {PARAFAC}-based technique for detection and
  localization of multiple targets in a {MIMO} radar system,'' in \emph{Proc.\
  {ICASSP}}, 2009, pp. 2077--2080.

\bibitem{Pfetsch14RIPCalcNPHard}
A.~M. Tillmann and M.~E. Pfetsch, ``The computational complexity of the
  restricted isometry property, the nullspace property, and related concepts in
  compressed sensing,'' \emph{{IEEE} Trans. Inf. Theory}, vol.~60, no.~2, pp.
  1248--1259, Feb. 2014.

\bibitem{HornAndJohnson}
R.~A. Horn and C.~R. Johnson, Eds., \emph{Matrix Analysis}.\hskip 1em plus
  0.5em minus 0.4em\relax Cambridge University Press, 1986.

\bibitem{SadeghJokar10KCS}
S.~Jokar, ``Sparse recovery and {K}ronecker products,'' in \emph{44th Annual
  Conference on Information Sciences and Systems}, Mar. 2010, pp. 1--4.

\bibitem{BhaskaraMoitra14SmoothedAna}
A.~Bhaskara, M.~Charikar, A.~Moitra, and A.~Vijayaraghavan, ``Smoothed analysis
  of tensor decompositions,'' in \emph{Proceedings of the 46th Annual ACM
  Symposium on Theory of Computing}, 2014, pp. 594--603.

\bibitem{AndersonBelinGoyal14}
J.~Anderson, M.~Belkin, N.~Goyal, L.~Rademacher, and J.~R. Voss, ``The more,
  the merrier: the blessing of dimensionality for learning large {G}aussian
  mixtures,'' in \emph{27th Annual Conference on Learning Theory}, Jun. 2014,
  pp. 1135--1164.

\bibitem{Sidiropoulos2000Kruskal}
N.~D. Sidiropoulos and R.~Bro, ``On the uniqueness of multilinear decomposition
  of {N}-way arrays,'' \emph{Journal of Chemometrics}, vol.~14, no.~3, pp.
  229--239, 2000.

\bibitem{KoiranZouzias14RIPCert}
P.~Koiran and A.~Zouzias, ``Hidden cliques and the certification of the
  restricted isometry property,'' \emph{{IEEE} Trans. Inf. Theory}, vol.~60,
  no.~8, pp. 4999--5006, Aug. 2014.

\bibitem{Khanna17RDCMP}
S.~Khanna and C.~R. Murthy, ``R\'{e}nyi divergence based covariance matching
  pursuit of joint sparse support,'' in \emph{Sig.\ Process.\ Advances in
  Wireless Commun.\ (SPAWC)}, Jul. 2017, pp. 1 -- 5.

\bibitem{Vershynin10RMT}
R.~Vershynin, ``Introduction to the non-asymptotic analysis of random
  matrices,'' 2010.

\bibitem{Visick2000HadamardAsSubMat}
G.~Visick, ``A quantitative version of the observation that the {H}adamard
  product is a principal submatrix of the {K}ronecker product,'' \emph{Linear
  Algebra and its Applications}, vol. 304, no. 1–3, pp. 45 -- 68, 2000.

\bibitem{horn94}
R.~A. Horn and C.~R. Johnson, \emph{Topics in Matrix Analysis}.\hskip 1em plus
  0.5em minus 0.4em\relax Cambridge; New York: Cambridge University Press,
  1994.

\bibitem{MondPecaric1998}
B.~Mond and J.~E. Pe\u{c}ari\'{c}, ``Inequalities for the {H}adamard product of
  matrices,'' \emph{SIAM J. Matrix Anal. Appl.}, vol.~19, no.~1, pp. 66--70,
  Jan. 1998.

\bibitem{Image26}
H.~J. Werner, ``Hadamard product of square roots of correlation matrices,''
  \emph{Image}, vol.~26, pp. 1--32, Apr. 2001.

\bibitem{MarshallOlkin64}
A.~W. Marshall and I.~Olkin, ``Reversal of the {L}yapunov, {H}{\"o}lder, and
  {M}inkowski inequalities and other extensions of the {K}antorovich
  inequality,'' \emph{Journal of Mathematical Analysis and Applications},
  vol.~8, no.~3, pp. 503 -- 514, 1964.

\bibitem{LiuNeudecker96}
S.~Liu and H.~Neudecker, ``Several matrix {K}antorovich-type inequalities,''
  \emph{Journal of Mathematical Analysis and Applications}, vol. 197, no.~1,
  pp. 23 -- 26, 1996.

\bibitem{LiuNeudecker97}
------, ``Kantorovich inequalities and efficiency comparisons for several
  classes of estimators in linear models,'' \emph{Statistica Neerlandica},
  vol.~51, no.~3, pp. 345--355, 1997.

\bibitem{LiuS2002EconometricTheory}
S.~Liu, ``On the {H}adamard product of square roots of correlation matrices,''
  \emph{Econometric Theory}, vol.~18, no.~4, p. 1007, Aug. 2002.

\bibitem{RaoNRao}
C.~R. Rao and M.~B. Rao, \emph{Matrix Algebra and Its Applications to
  Statistics and Econometrics}.\hskip 1em plus 0.5em minus 0.4em\relax
  Singapore: World Scientific, 1998.

\bibitem{adamczak2015}
R.~Adamczak, ``A note on the {H}anson-{W}right inequality for random vectors
  with dependencies,'' \emph{Electron. Commun. Probab.}, vol.~20, pp. 72--84,
  2015.

\bibitem{Rudelson13HansonWrightIneq}
M.~Rudelson and R.~Vershynin, ``Hanson-{W}right inequality and sub-gaussian
  concentration,'' \emph{Electron. Commun. Probab.}, vol.~18, pp. 82--91, 2013.

\end{thebibliography}

\ifdefined \SKIPTEMP
\begin{IEEEbiography}[{\includegraphics[width=1in, height=2.25in,clip,keepaspectratio]{Saurabh_Khanna_photo.eps}}]{Saurabh Khanna} received 
the B.\ Tech.\ degree in Electrical Engineering from the Indian Institute of Technology, Kanpur in 2007. 
From 2007 to 2016, he was with Texas Instruments, Bangalore working on firmware and algorithm design for WLAN transceivers, Global Navigation Satellite System (GNSS) based user localization and FMCW radars. He is currently pursuing Ph.\ D.\ degree in Electrical 
Communication Engineering at Indian Institute of Science, Bangalore, India. His research interests are in the areas of structured signal processing, inverse problems and statistical learning theory.  
\end{IEEEbiography}

\begin{IEEEbiography}[{\includegraphics[width=1in,height=2.25in,clip,keepaspectratio]{Chandra_Murthy_photo.eps}}]{Chandra R. Murthy} (S'03--M'06--SM'11) received 
the B.\ Tech. degree in Electrical Engineering from the Indian Institute of Technology, Madras in 1998, the M.\ S.\ and Ph.\ D.\ degrees in 
Electrical and Computer Engineering from Purdue University and the University of California, San Diego, in 2000 and 2006, respectively.
From 2000 to 2002, he worked as an engineer for Qualcomm Inc., where he worked on WCDMA baseband transceiver design and 802.11b baseband 
receivers. From Aug. 2006 to Aug. 2007, he worked as a staff engineer at Beceem Communications Inc.\ on advanced receiver architectures 
for the 802.16e Mobile WiMAX standard. In Sept. 2007, he joined the Department of Electrical Communication Engineering at the Indian 
Institute of Science, Bangalore, India, where he is currently working as an Associate Professor. 

His research interests are in the areas of energy harvesting communications, multiuser MIMO systems, and sparse signal recovery techniques 
applied to wireless communications. 
His paper won the best paper award in the communications track in the National Conference on Communications 2014. 
He has coauthored 47 journal and 82 conference papers.
He was an associate 
editor for the IEEE Signal Processing Letters during 2012-16. He is an elected member of the IEEE SPCOM Technical Committee for the 
years 2014-2016, and has been reelected for the years  2017-2019.
He is a past Chair of the IEEE Signal Processing Society, Bangalore Chapter, and is currently serving as 
an associate editor for the IEEE Transactions on Signal Processing, the Sadhana journal, and the IEEE Transactions on Communications.
\end{IEEEbiography}
\fi

\end{document}